\documentclass[conference]{sig-alternate}

\makeatletter
\def\ps@headings{%
\def\@oddhead{\mbox{}\scriptsize\rightmark \hfil \thepage}%
\def\@evenhead{\scriptsize\thepage \hfil \leftmark\mbox{}}%
\def\@oddfoot{}%
\def\@evenfoot{}}
\def\blfootnote{\xdef\@thefnmark{}\@footnotetext}
\makeatother
\usepackage{graphicx}
\usepackage{dcolumn}
\usepackage{bm}
\usepackage{epsfig}
\usepackage{psfrag}
\usepackage{subfigure}
\usepackage{color}
\usepackage{amsmath, amssymb, latexsym,amsfonts,epsfig, graphicx, verbatim}
\usepackage[nospace]{cite}
\usepackage{algorithmic, algorithm}
\usepackage{verbatim}
\usepackage{xspace}

\newtheorem{observation}{Observation}

\newcommand{\back}{\!\!\!}
\newcommand{\eq}{\!=\!}
\newcommand{\m}{\!-\!}
\newcommand{\Prob}{\mathbb{P}}
\newcommand{\Mean}{\mathbb{E}}
\newcommand{\Var}{\mathbb{V}}
\newcommand{\eqn}[1]{Eq.(#1)}
\newcommand{\vol}{\textrm{vol}}
\newcommand{\w}{\textrm{w}}
\newcommand{\z}{\textrm{w}}

\newcommand{\lowhat}[1]{\rlap{\raisebox{-0.5ex}{$\hat{\phantom{#1}}$}}#1}
\newcommand{\lowtilde}[1]{\rlap{\raisebox{-0.5ex}{$\tilde{\phantom{#1}}$}}#1}
\newcommand{\prop}{{\scriptscriptstyle\textrm{prop}}}
\newcommand{\opt}{{\scriptscriptstyle\textrm{opt}}}
\newcommand{\heur}{S-WRW\xspace}
\newcommand{\UIS}{{\scriptscriptstyle\textrm{UIS}}}
\newcommand{\WIS}{{\scriptscriptstyle\textrm{WIS}}}
\newcommand{\RW}{{\scriptscriptstyle\textrm{RW}}}
\newcommand{\MHRW}{{\scriptscriptstyle\textrm{MHRW}}}
\newcommand{\WRW}{{\scriptscriptstyle\textrm{WRW}}}
\newcommand{\NRMSE}{{\small\textrm{NRMSE}\xspace}}

\newcommand{\sss}[1]{{\scriptscriptstyle\textrm{#1}}}
\newcommand{\MySmall}[1]{{\scriptsize{#1}}}

\newcommand{\mybox}[1]{\noindent\frame{\parbox[c]{\columnwidth}{\begin{center}{\vspace{-0.28cm}#1}\vspace{-0.3cm}\end{center}}}}

\newcommand{\widearrow}{
\vspace{-0.07cm}
\includegraphics[width=0.11\textwidth]{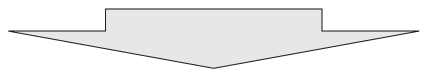}\\
\vspace{0.0cm}
}
\newcommand{\Sec}[1]{Sec.~\ref{#1}}
\newcommand{\Fig}[1]{Fig.~\ref{#1}}

\newcommand{\note}[1]{[\textcolor{red}{\textit{#1}}]}

\newcommand{\est}[1]{\widehat{#1}}

\newcommand{\ie}{{\em i.e., }}
\newcommand{\eg}{{\em e.g., }}

\newfont{\smallttlfnt}{phvb8t at 14pt}

\begin{document}

\conferenceinfo{SIGMETRICS'11,} {June 7--11, 2011, San Jose, California, USA.} 
\CopyrightYear{2011} 
\crdata{978-1-4503-0262-3/11/06} 
\clubpenalty=10000 
\widowpenalty = 10000

\title{Walking on a Graph with a Magnifying Glass\\
{\smallttlfnt Stratified Sampling via Weighted Random Walks}
}

\numberofauthors{4} 

\author{
Maciej Kurant, Minas Gjoka, Carter T. Butts, Athina Markopoulou\\
\affaddr{University of California, Irvine}\\
\email{\{mkurant, mgjoka, buttsc, athina\}@uci.edu}
}

\maketitle

\blfootnote{\!\!\!* This is an extended version of a paper with the same title presented at \emph{SIGMETRICS'11}. This work was supported by SNF grant PBELP2-130871, Switzerland, and by the NSF CDI Award 1028394, USA. 
\vspace{-1cm}
}

\begin{abstract}
Our objective is to sample the node set of a large unknown graph via crawling, to accurately estimate a given metric of interest.
We design a random walk on an appropriately defined weighted graph  
that achieves high efficiency by preferentially crawling those nodes and edges that convey greater information regarding the target metric.
Our approach begins by employing the theory of stratification 
to find optimal node weights, for a given estimation problem, under an independence sampler. While optimal under independence sampling, these weights may be impractical 
under graph crawling due to constraints arising from the structure of the graph.
Therefore, the edge weights for our random walk should be chosen so as to lead to an equilibrium distribution that strikes a balance between approximating the optimal weights under an independence sampler and achieving fast convergence.
We propose a heuristic approach (stratified weighted random walk, or S-WRW) that achieves this goal, while using only limited information about the graph structure and the node properties.
We evaluate our technique in simulation, and experimentally, by collecting a sample of Facebook college users.
We show that S-WRW requires \mbox{13-15} times fewer samples than the simple re-weighted random walk (RW) to achieve the same estimation accuracy for a range of metrics.

\end{abstract}

\section{Introduction}\label{Introduction}

\begin{figure*}
		\psfrag{A}[c]{\small{\textbf{(a)} Population}}
		\psfrag{B}[c]{\small{\textbf{(b)} WIS weights $\pi^\WIS$}}
		\psfrag{C}[c]{\small{\textbf{(c)} Social graph $G$}}
		\psfrag{D}[c]{\small{\textbf{(d)} $\pi^\WIS$ applied to WRW}}
		\psfrag{E}[c]{\small{\textbf{(e)} \heur weights}}
    \includegraphics[width=1\textwidth]{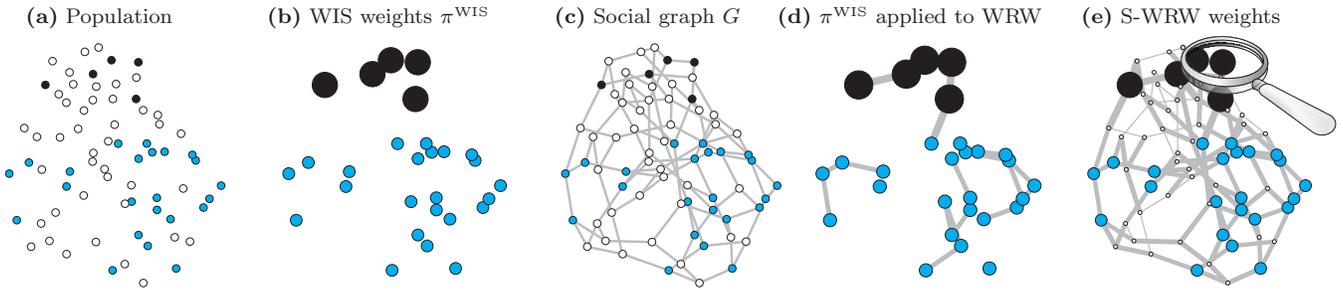}
    \vspace{-0.6cm}
\caption{Illustrative example. \ 
%
Our goal is to compare the blue and black subpopulations (\eg with respect to their median income) in population~(a). 
Optimal independence sampler, WIS~(b), over-samples the black nodes, under-samples the blue nodes, and completely skips the white nodes. 
A naive crawling approach, RW~(c), samples many irrelevant white nodes. 
WRW that enforces WIS-optimal probabilities may result in poor or no convergence~(d). 
\heur~(e) strikes a balance between the optimality of WIS and fast convergence.
}
\label{fig:illustration_magnifying_glass}
\vspace{-10pt}
\end{figure*}

Many types of online networks, such as online social networks (OSNs), Peer-to-Peer (P2P) networks,  or the World Wide Web (WWW), are measured and studied today via sampling techniques.
This is due to several reasons.
First, such graphs are typically too large to measure in their entirety, and it is desirable to be able to study them based on a small but representative  sample. 
Second, the information pertaining to these networks is often hard to obtain. For example, OSN service providers have access to all information in their user base, but rarely make this information publicly available.

There are many ways a graph can be sampled, \eg by sampling nodes, edges, paths, or other substructures~\cite{Leskovec2006_sampling_from_large_graphs,Kolaczyk2009}.
Depending on our measurement goal, the elements with different properties may have {\em different importance} and should be sampled with a different probability. 
For example, Fig.~\ref{fig:illustration_magnifying_glass}(a) depicts the world's population, with residents of China (1.3B people) represented by blue nodes,  of the Vatican (800 people) by black nodes, and all other nationalities represented by white nodes. Assume that we want to compare the median income in China and Vatican.
Taking a uniform sample of size 100 from the entire world's population is ineffective, because most of the samples will come from countries other than China and Vatican.
Even restricting our sample to the union of China and Vatican will not help much, as our sample is unlikely to include any Vatican resident. In contrast, uniformly sampling 50 Chinese and 50 Vaticanese residents 
would be much more accurate with the same sampling budget.

This type of problem has been widely studied in the statistical and survey sampling literature. A commonly used approach is \emph{stratified sampling}~\cite{Neyman1934,Cochran1977,Lohr2009}, 
where nodes (\eg people) are partitioned into a set of non-overlapping \emph{categories} (or strata). The objective is then to decide how many independent draws to take from each category, so as to minimize the uncertainty of the resulting measurement.
This effect can be achieved in expectation by a weighted independence sampler (WIS) with appropriately chosen sampling probabilities~$\pi^\WIS$.
In our example, WIS samples Vatican residents with much higher probabilities than Chinese ones, and avoids completely the rest of the world, as illustrated in Fig.~\ref{fig:illustration_magnifying_glass}(b).

However, WIS, as every independence sampler, requires a sampling frame, \ie a list of all elements we can sample from (\eg a list of all Facebook users).
This information is typically not available in today's online networks. 
A feasible alternative is 
\emph{crawling} (also known as exploration or link-trace sampling). 
It is a graph sampling technique in which we can see the neighbors of already sampled users and make a decision on which users to visit next.


In this paper, we study how to perform stratified sampling through graph crawling.
We illustrate the key idea and some of the challenges in Fig.~\ref{fig:illustration_magnifying_glass}.
Fig.~\ref{fig:illustration_magnifying_glass}(c) depicts a social network that connects the world's population.
A simple random walk (RW) visits every node with frequency proportional to its degree, which is reflected by the node size. In this particular example, for a simplicity of illustration, all nodes have the same degree equal to~3.
As a result, RW is equivalent to the uniform sample of the world's population, and faces exactly the same problems of wasting resources, by sampling all nodes with the same probability.

We address these problems by appropriately setting the edge weights and then performing a random walk on the weighted graph, which we refer to as \emph{weighted random walk} (WRW). One goal in setting the weights is to mimic the WIS-optimal sampling probabilities~$\pi^\WIS$ shown in Fig.~\ref{fig:illustration_magnifying_glass}(b).
However, such a WRW might perform poorly due to potentially slow mixing. In our example, it will not even converge because the underlying weighted graph is disconnected, as shown in Fig.~\ref{fig:illustration_magnifying_glass}(d). 
Therefore, the edge weights under WRW (which determine the equilibrium distribution~$\pi^\WRW$) should be chosen in a way that strikes a balance between the optimality of~$\pi^\WIS$ and fast convergence.

We propose {\em Stratified Weighted Random Walk} (\heur), a practical heuristic that effectively strikes such a balance. We refer to our approach as ``walking on the graph with a magnifying glass'', because \heur
over-samples more relevant parts of the graph and under-samples less relevant ones.
In our example, \heur results in the graph presented in Fig.~\ref{fig:illustration_magnifying_glass}(e).
The only information required by \heur are the categories of neighbors of every visited node, 
which is typically available in crawlable online networks, such as Facebook.
\heur uses two natural and easy-to-interpret parameters, namely: (i) $\tilde{f}_\ominus$, which controls the fraction of samples from irrelevant categories and (ii)~$\gamma$, which is the maximal resolution of our magnifying glass, with respect to the largest relevant category.

The main contributions of this paper are the following.
\begin{itemize}
\item We propose to improve the efficiency of crawling-based graph sampling methods, by performing a stratified weighted random walk  that takes into account not only the graph structure but also the node properties that are relevant to the measurement goal.
\item We design and evaluate \heur, a practical heuristic that sets the edge weights and operates with limited information.
\item As a case study, we apply \heur to sample Facebook and estimate the sizes of colleges. We show that \heur requires 13-15 times fewer samples than a simple random walk for the same estimation accuracy.
\end{itemize}

The outline of the rest of the paper is as follows. Section~2 summarizes the most popular graph sampling techniques, including sampling by exploration.
Section~3 presents classical stratified sampling. 
Section~4 combines stratified sampling with graph exploration, presenting a unified WRW approach that takes into account both network structure and node properties; various trade-offs and practical issues are discussed and an efficient heuristic (\heur) is proposed based on the insights. 
Section~5 presents simulation results. 
Section~6 presents an implementation of \heur for the problem of estimating the college friendship graph on Facebook.
Section~7 presents related work. Section~8 concludes the paper. 

\section{Sampling techniques}

\subsection{Notation}
We consider an undirected, static,\footnote{Sampling dynamic graphs is currently an active research area~\cite{Stutzbach2006-unbiased-p2p, Rasti09-RDS, Willinger09-OSN_Research}, but out of the scope of this paper.}
 graph $G=(V,E)$, with $N\eq|V|$~nodes and $|E|$~edges.
For a node $v\in V$, denote by $\deg(v)$ its degree, and by $\mathcal{N}(v)\subset V$ the list of neighbors of~$v$.
A graph~$G$ can be weighted.
We denote by $\w(u,v)$ the weight of edge $\{u,v\}\in E$, and by
\begin{equation}
\label{eq:w_u}
\w(u) = \sum_{v\in \mathcal{N}(u)} \w(u,v)
\end{equation}
the weight of node~$u\in V$.
For any set of nodes $A\subseteq V$, we define its volume $\vol(A)$ and weight $\w(A)$, respectively, as
\begin{equation}\label{eq:volume}
\vol(A) = \sum_{v\in A} \deg(v) \quad \textrm{ and } \quad \w(A) = \sum_{v\in A} \w(v).
\end{equation}
We will often use
\begin{equation}\label{eq:relative fractions}
f_A = \frac{|A|}{|V|}  \quad \textrm{ and } \quad f^\sss{vol}_A = \frac{\vol(A)}{\vol(V)}
\end{equation}
to denote the relative size of $A$ in terms of the number of nodes and the volumes, respectively.

\smallskip\noindent\textbf{Sampling.}
We collect a sample $S\subseteq V$ of $n\eq|S|$ nodes. $S$~may contain multiple copies of the same node, \ie the sampling is with replacement.
In this section, we briefly review the techniques for sampling nodes from graph~$G$. We also present the weighted random walk (WRW) which is the basic building block for our approach.

\subsection{Independence Sampling}

\smallskip\noindent\textbf{Uniform Independence Sampling (UIS)}
 samples the nodes directly from the set~$V$, with replacements, uniformly and independently at random, \ie with probability
\begin{equation}\label{eq:pi_UIS}
	\pi^\UIS(v)\ =\ \frac{1}{N} \qquad\textrm{ for every }v\in V.
\end{equation}

\smallskip\noindent\textbf{Weighted Independence Sampling (WIS)} is a weighted version of UIS. WIS samples the nodes directly from the set~$V$, with replacements, independently at random, but with probabilities proportional to node weights~$\z(v)$:
\begin{equation}\label{eq:pi_WIS}
	\pi^\WIS(v)\ =\ \frac{\z(v)}{\sum_{u\in V} \z(u)}.
\end{equation}

\smallskip
\noindent In general, UIS and WIS are not possible in online networks because of the lack of sampling frame. For example, the list of all user IDs may not be publicly available, or the user ID space may be too sparsely allocated.
Nevertheless, we present them as baseline for comparison with the random walks.

\subsection{Sampling via Crawling\label{sec:walks}}
In contrast to independence sampling, the crawling techniques are possible in many online networks, and are therefore the main focus of this paper.

\smallskip\noindent\textbf{Simple Random Walk (RW)} \cite{Lovasz93} selects the next-hop node~$v$ uniformly at random among the neighbors of the current node~$u$.
In a connected and aperiodic graph, the probability of being at the particular node~$v$ converges to the stationary distribution
\begin{equation}\label{eq:pi_RW}
	\pi^\RW(v)\ =\ \frac{\deg(v)}{2\cdot|E|}.
\end{equation}

\smallskip\noindent\textbf{Metropolis-Hastings Random Walk (MHRW)} is an application of the Metropolis-Hastings algorithm \cite{Metropolis1953} that modifies the transition probabilities to converge to a desired stationary distribution.
For example, we can achieve the uniform stationary distribution
\begin{equation}\label{eq:pi_MHRW}
	\pi^\MHRW(v)\ =\ \frac{1}{N}
\end{equation}
by randomly selecting a neighbor~$v$ of the current node~$u$ and moving there with probability $\min(1,\frac{\deg(u)}{\deg(v)})$.
%
However, it was shown in \cite{Rasti09-RDS, Gjoka2010} that RW (after re-weighting, as in Section~\ref{subsec:Correcting for the bias in RW, WRW and WIS}) outperforms MHRW for most applications. We therefore restrict our attention to comparing against RW.

\smallskip\noindent\textbf{Weighted Random Walk (WRW)} is RW on a weighted graph~\cite{AldousBookInPreparation}. At node~$u$, WRW chooses the edge $\{u,v\}$ to follow with probability $P_{u,v}$ proportional to the weight $\w(u,v)\geq 0$ of this edge, \ie
\begin{equation}\label{eq:P_u_v in WRW}
P_{u,v}=\frac{\w(u,v)}{\sum_{v'\in \mathcal{N}(u)} \w(u,v')}.
\end{equation}
The stationary distribution of WRW is:
\begin{equation}\label{eq:pi_WRW}
	\pi^\WRW(v)\ =\ \frac{\w(v)}{\sum_{u\in V} \w(u)}.
\end{equation}
WRW is the basic building block of our design. In the next sections, we show how to choose weights for a specific estimation problem.

\smallskip\noindent\textbf{Graph Traversals (BFS, DFS, RDS, ...)} is a family of crawling techniques where no node is sampled more than once. 
Because traversals introduce a generally unknown bias (see~\Sec{sec:related}), 
we do not consider them in this paper.

\subsection{Correcting the bias}
\label{subsec:Correcting for the bias in RW, WRW and WIS}
RW, WRW, and WIS all produce biased (nonuniform) node samples. 
But their bias is known and therefore can be corrected by an appropriate re-weighting of the measured values.
This can be done using the Hansen-Hurwitz estimator~\cite{HansenHurwitz1943} as first shown in \cite{Salganik2004, VolzHeckathorn08} for random walks and also used in \cite{Rasti09-RDS}.
Let every node $v\in V$ carry a value $x(v)$. We can estimate the population total $x_\sss{tot} = \sum_v x(v)$ by
%
\begin{equation}\label{f_tot}
	\hat{x}_\sss{tot} = \frac{1}{n}\sum_{v\in S} \frac{x(v)}{\pi(v)},
\end{equation}
where 
$\pi(v)$ is the sampling probability of node~$v$ in the stationary distribution.  
%
In practice, we usually know~$\pi(v)$, and thus~$\hat{x}_\sss{tot}$,  only up to a constant, \ie we know the (non-normalized) weights~$\w(v)$. 
This problem disappears when we estimate the population mean $x_\sss{av}=\sum_v x(v)/N$ as
%
\begin{equation}\label{f_tot}  
	\hat{x}_\sss{av} \ =\ \frac{\sum_{v\in S} \frac{x(v)}{\pi(v)}}{\sum_{v\in S} \frac{1}{\pi(v)}} \ =\  \frac{\sum_{v\in S} \frac{x(v)}{\w(v)}}{\sum_{v\in S} \frac{1}{\w(v)}}.
\end{equation}
For example, for $x(v)\eq1$ if $\deg(v)\eq k$ (and $x(v)\eq 0$ otherwise), $\hat{x}_\sss{av}(k)$  estimates 
the node degree distribution in~$G$. 



\smallskip
All the results in this paper are presented \emph{after this re-weighting} step, whenever necessary.


%

\section{Stratified Sampling}\label{sec:Node weight setting under WIS}

In Sec.~\ref{Introduction}, we argued that in order to compare the median income of residents of China and Vatican we should take 50 random samples from each of these two countries, rather than taking 100 UIS samples from China and Vatican together (or, even worse, from the world's population).
This problem naturally arises in the field of survey sampling. The most common solution is \emph{stratified sampling}~\cite{Neyman1934,Cochran1977,Lohr2009}, where nodes~$V$ are partitioned into a set~$\mathcal{C}$ of non-overlapping node categories (or ``strata''), with $\bigcup_{C\in\mathcal{C}}C = V$. Next, we select uniformly at random~$n_i$ nodes from category~$C_i$. We are free to choose the 
allocation $(n_1,n_2,\ldots,n_{|\mathcal{C}|})$, as long as we respect the total budget of samples~$n\eq\sum_i n_i$. 

Under \emph{proportional allocation}~\cite{Lohr2009} (or ``prop') we use $n_i\propto |C_i|$, \ie
\begin{equation}\label{eq:n_i proportional allocation}
 n_i^\prop\ =\ |C_i|\cdot n/N.
\end{equation}

Another possibility is to do an {\em optimal} allocation (or ``opt'') that minimizes the variance $\Var$ of our estimator for the specific problem of interest.
 For example, assume that every node $v\in V$ carries a value~$x(v)$, and we may want to estimate the mean of $x$ 
in various scenarios, as discussed below.

\subsection{Examples of Stratified Sampling Problems}

\subsubsection{Estimating the mean across the entire $V$}\label{subsec:The mean across entire V}
A classic application of stratification is to better estimate the population mean~$\mu$, given several groups (strata) of different properties (\eg variances).
%
Given $n_i$ samples from category $C_i$, we can estimate the mean $\mu_i = \frac{1}{|C_i|}\sum_{v\in C_i} x(v)$ over category~$C_i$ by
\begin{equation}\label{eq:hat mu i}
	\hat{\mu}_i = \frac{1}{n_i}\sum_{v\in S\cap C_i} x(v) \qquad \textrm{ with } \qquad \Var(\hat{\mu}_i)=\frac{\sigma_i^2}{n_i},
\end{equation}
where $\Var(\hat{\mu}_i)$ is the variance of this estimator
and $\sigma_i^2$ is the variance of population~$C_i$.
We can estimate population mean~$\mu$ by a weighted average over all~$\hat{\mu}_i$s~\cite{Lohr2009}, \ie
$$\hat{\mu} = \sum_i \frac{|C_i|}{N} \cdot \hat{\mu}_i \qquad \textrm{ with } \qquad \Var(\hat{\mu}) = \sum_i \frac{(|C_i|)^2\cdot\sigma_i^2}{N^2\cdot n_i}.$$ 
Under proportional allocation (\eqn{\ref{eq:n_i proportional allocation}}), this boils down to
$\Var(\hat{\mu}^\prop)\ =\  \frac{1}{N\cdot n} \ \sum_i |C_i|\cdot\sigma_i^2$.
However, we can apply Lagrange multipliers to find that $\Var(\hat{\mu})$ is minimized when
\begin{equation}\label{eq: optimal n_i for The mean across entire V}
	n_i^\opt = \frac{|C_i|\cdot\sigma_i}{\sum_j |C_j|\cdot\sigma_j}\cdot n.
\end{equation}
This solution is sometimes called `Neyman allocation'~\cite{Neyman1934}.
This gives us the variance under optimal allocation
$\Var(\hat{\mu}^\opt)\ =\ \frac{1}{N^2\cdot n}\  \left(\sum_i |C_i|\cdot\sigma_i\right)^2$.

The variances $\Var(\hat{\mu}^\prop)$ and $\Var(\hat{\mu}^\opt)$ are  measures of the performance of proportional and optimal allocation, respectively.
%
In order to make their practical interpretation easier, we also show how these variances translate into sample lengths.
We define as \emph{gain}~$\alpha$ of `opt' over `prop' the number of times `prop' must be longer than `opt' in order to achieve the same variance
$$\textrm{gain }\ \alpha \ =\  \frac{n^\prop}{n^\opt}, \ \textrm{ subject to }\ \Var^\prop\eq\Var^\opt.$$
In that case, the gain is
\begin{equation}\label{eq: gain under the mean across entire V}
	\alpha\ \ =\ \ N\cdot \frac{\sum_i |C_i|\cdot\sigma_i^2}{\left(\sum_i |C_i|\cdot\sigma_i\right)^2} \qquad (\geq 1).
\end{equation}
Notice that this gain does not depend on the sample budget $n$.
The gain is one of the main metrics we will use in the evaluation sections to assess the efficiency of our technique
compared to the random walk.


\subsubsection{Highest precision for all categories}\label{subsec: Highest precision for all}
If we are equally interested in each category, we might want the same (highest possible) precision of estimating~$\mu_i$ for all categories~$C_i$. In this case, the metric to minimize is
$	\Var_{\max}\ =\ \max_i\left\{\Var(\hat{\mu}_i)\right\}\ = \max_i\left\{\frac{\sigma_i^2}{n_i}\right\}.  $
Under proportional allocation, this translates to
$\Var_{\max}^\prop \ =  \frac{N}{n} \max_i \frac{\sigma_i^2}{|C_i|}$.
But the optimal $n_i$, which makes $\Var(\hat{\mu}_i)$ equal for all $i$, is
\begin{equation}\label{eq: optimal n_i for max}
	n_i^\opt = \frac{\sigma^2_i}{\sum_j \sigma^2_j}\cdot n.
\end{equation}
Consequently, $\Var_{\max}^\opt \ = \ \frac{\sum_i \sigma_i^2}{n},$
which leads to gain
\begin{equation}\label{eq:gain under max}
	\alpha = \frac{\max_i \left\{\frac{N}{|C_i|} \sigma_i^2\right\}}  {\sum_i \sigma_i^2}\quad (\geq 1).
\end{equation}

\subsubsection{Smallest sum of variances across categories}
Even if we are interested in all categories, an alternative objective is to maximize the \emph{average} precision of category pair comparisons (see Sec.~5A.13 in~\cite{Cochran1977}), which is equivalent to minimizing the sum
$	\Var_\Sigma=\sum_i \Var(\hat{\mu}_i)=\sum_i \frac{\sigma_i^2}{n_i}.$
%
In this case, proportional allocation achieves
$\Var_\Sigma^\prop = \frac{N}{n} \sum_i \frac{\sigma_i^2}{|C_i|}$.
while, using Lagrange multipliers we get
\begin{equation}\label{eq: optimal n_i for sum}
	n_i^\opt = \frac{\sigma_i}{\sum_j \sigma_j}\cdot n \qquad \textrm{ and } \qquad \Var_\Sigma^\opt = \frac{\left(\sum_i \sigma_i\right)^2}{n},
\end{equation}
which leads to gain
\begin{equation}\label{eq:gain under smallest sum of variances}
	\alpha \ =\ \frac{\sum_i \frac{N}{|C_i|} \sigma_i^2}  {\left(\sum_i \sigma_i\right)^2}\quad (\geq 1).
\end{equation}


\subsubsection{Relative sizes of node categories}\label{subsec: Relative sizes of node categories}
Stratified sampling assumes that we know the sizes~$|C_i|$ of node categories. In some applications, however, these sizes are unknown and among the values we need to estimate as well (\eg by using UIS or WIS).
We show in Appendix~C (for $|\mathcal{C}|\eq 2$) that
the optimal sample allocation and the corresponding gain~$\alpha$ of WIS over UIS are respectively
\begin{equation}\label{eq: Relative sizes of node categories}
	n_i^\WIS = \frac{1}{|\mathcal{C}|}\cdot n \quad \textrm{ and }\quad \alpha\ =\ \frac{N^2}{4|C_1|\cdot|C_2|}.
\end{equation}


\subsubsection{Irrelevant category $C_\ominus$ (aggregated)} \label{subsec:Irrelevant category}
In many practical cases, we may want to measure some (but not all) node categories.
E.g., in Fig.~\ref{fig:illustration_magnifying_glass}, we are interested in blue and black nodes, but not in white ones. Similarly, in our Facebook study in Section~\ref{Applications} we are only interested in self-declared college students, which accounts for only 3.5\% of all users.
We group all categories not covered by our measurement objective as a single \emph{irrelevant category} $C_\ominus\in \mathcal{C}$, and we set~$n^\opt_\ominus=0$.
In contrast, $n_\ominus^\prop\ =\ |C_\ominus|\cdot n/N$. As a result, under `opt' we have $N / (N\m |C_\ominus|)$ times more useful samples than under `prop'. Now, if we allocate optimally all these useful samples between the relevant categories $\mathcal{C}\setminus\{C_\ominus\}$,
the gain~$\alpha$ becomes
\begin{equation}\label{eq:alpha hybrid}
	\alpha \ \ = \ \ \frac{N}{N-|C_\ominus|}\ \cdot\ \alpha(\mathcal{C}\setminus\{C_\ominus\}),
\end{equation}
where $\alpha(\mathcal{C}\setminus\{C_\ominus\})$ is the gain
(\ref{eq: gain under the mean across entire V}), (\ref{eq:gain under max}), (\ref{eq:gain under smallest sum of variances}) or (\ref{eq: Relative sizes of node categories}),
depending on the metric, calculated only within categories $\mathcal{C}\setminus\{C_\ominus\}$.

In other words, gain~$\alpha$ is now composed of two factors: (i)~gain in avoiding irrelevant categories, and (ii)~gain in optimal allocation of samples among the relevant categories.

\subsubsection{Practical Guideline}
Let us look at the optimal weights in the above scenarios, when all $\sigma_i=\sigma$ are the same. This is a reasonable working assumption in many practical settings, since we typically do not have prior estimates of~$\sigma_i$.
With this simplification, \eqn{\ref{eq: optimal n_i for The mean across entire V}} becomes
$$n_i^\opt\ =\ \frac{|C_i|}{N}\cdot n\ =\ n_i^\prop.$$
In contrast, \eqn{\ref{eq: optimal n_i for max}}, \eqn{\ref{eq: optimal n_i for sum}} and \eqn{\ref{eq: Relative sizes of node categories}} get simplified to
$$ n_i^\opt\ =\ \frac{1}{|\mathcal{C}|}\cdot n.$$
In conclusion, if we are interested in comparing the node categories with respect to some properties (\eg average node degree, category size), rather than estimating a property across the entire population, we should take an \emph{equal number of samples from every relevant category}.

\section{Edge weight setting under WRW}\label{sec:Edge weight setting under WRW}

In the previous section, we studied the optimal sample allocation under (independence) stratified sampling.
However, independence node sampling is typically impossible in large online graphs, while crawling the graph is
a natural, available exploration primitive. In this section, we show how to perform a weighted random walk (WRW) which approximates the stratified sampling of the previous section. 
We can formulate the general problem as follows:

\emph{Given a measurement objective, error metric and sampling budget~$|S|\eq n$, 
set the edge weights in graph~$G$ such that the WRW measurement error is minimized.}

Although we are able to solve this problem analytically for some specific and fully known topologies, it is not obvious how to address it in general, especially under a limited knowledge of~$G$. 
Instead, in this paper, we propose \heur, a heuristic to set the edge weights. \heur starts from a solution optimal under WIS, and takes into account practical issues that arise in graph exploration. 
Once the  weights are set, we simply perform WRW as described in Section \ref{sec:walks} and collect samples.

\subsection{Preliminaries}

\subsubsection{Category-level granularity}
One can think of the problem in two levels of granularity: the original graph $G\eq (V,E)$ and the {\em category graph} $G^C \eq (\mathcal{C}, E^C)$. In~$G^C$, nodes represent categories, and every undirected edge $\{C_1,C_2\}\in E^C$ represents the corresponding non-empty set of edges $E_{C_1,C_2}\subset E$ in the original graph~$G$, \ie
$$E_{C_1,C_2}\ = \{\{u,v\}\in E:\ u\in C_1 \textrm{ and }  v\in C_2\} \neq \emptyset.$$

In our approach, we move from the finer granularity of $G$ to the coarser granularity of $G^C$. 
This means that we are interested in collecting, say, $n_i$ samples from category $C_i$, but we do not control how these $n_i$ nodes are collected (\ie with what individual sampling probabilities). 

The rationale for that simplification is twofold. 
From a theoretical point of view, categories are exactly the properties of interest in the estimation problems we consider. 
From a practical point of view, it is relatively easy to obtain or infer information about categories, as we show \eg in~\Sec{subsec:Estimation of Category Volumes}.

\subsubsection{Stratification in expectation}
Ideally, we would like to enforce strictly stratified sampling. 
However, when we use crawling instead of independence sampling, sampling exactly~$n_i$ nodes from category~$C_i$ (and no other nodes) is possible only by discarding observations. 
It is thus more natural to frame the problem in terms of the probability mass placed on each category in equilibrium. 
This can be achieved by making the weight $\w(C_i)$ of each category proportional to the desired number~$n_i$ of samples, \ie 
\begin{equation}\label{w(C_i) propto n_i}
	\w(C_i)\ \propto\ n_i.
\end{equation}
As a result, we draw $n_i$ samples from $C_i$ \emph{in expectation}. 


%
%

\subsubsection{Main guideline}
As the main guideline, \heur tries to realize the category weights $\w^\WIS(C_i)$ that are optimal under WIS. 
There are many edge weight settings in~$G$ that achieve~$\w^\WIS(C_i)$. 
In our implementation, we observe that~$\vol(C_i)$ counts the number of edges incident on nodes of~$C_i$. Consequently, if for every category~$C_i$ we set in~$G$ the weights of all edges incident on nodes in~$C_i$ to
\begin{equation}\label{eq:w_e(C_i)}
	\w_e(C_i)\ =\ \frac{\w^\WIS(C_i)}{\vol(C_i)}.
\end{equation}
then weight $\w^\WIS(C_i)$ are achieved.\footnote{There exist many other edge weight assignments that lead to $\w^\WIS(C_i)$. \eqn{\ref{eq:w_e(C_i)}} has the advantage of distributing the weights evenly across all~$\vol(C_i)$ edges.}
This simple observation is central to the \heur heuristic. 

In order to apply \eqn{\ref{eq:w_e(C_i)}}, we first have to calculate or estimate its terms $\vol(C_i)$ and $\w^\WIS(C_i)$.\footnote{\label{footnote:constant factor}In fact, we need to know $\w_e(C_i)$ in \eqn{\ref{eq:w_e(C_i)}} only \emph{up to a constant factor}, because these factors cancel out in the calculation of transition probabilities of WRW in \eqn{\ref{eq:P_u_v in WRW}}. Consequently, the same applies to $\vol(C_i)$ and $\w^\WIS(C_i)$.}
Below, we show how to do it in Step~1 and 2, respectively. Next, in Steps~3-5, we show how to modify these terms to account for practical problems arising mainly from the underlying graph structure.

\begin{figure}
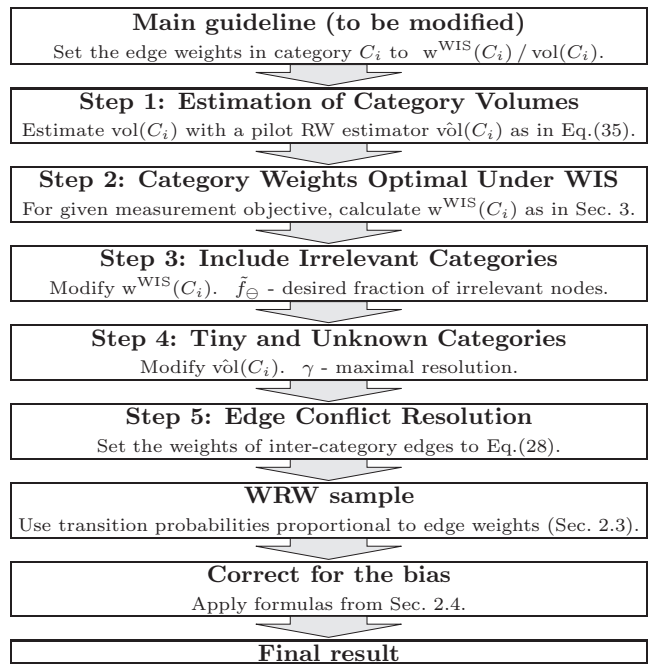

\begin{center}
\mybox{\textbf{Main guideline (to be modified)}\\
\MySmall{Set the edge weights in category $C_i$ to \ $\w^\textrm{WIS}(C_i)\, /\, \vol(C_i).$}}\\
\widearrow
\mybox{\textbf{Step 1: Estimation of Category Volumes}\\ 
\MySmall{Estimate $\vol(C_i)$ with a pilot RW estimator $\lowhat{\vol}(C_i)$ as in \eqn{\ref{vol c 2 RW}}.}}\\
\widearrow
\mybox{\textbf{Step 2: Category Weights Optimal Under WIS}\\
\MySmall{For given measurement objective, calculate $\w^\WIS(C_i)$ as in \Sec{sec:Node weight setting under WIS}.}}\\
\widearrow
\mybox{\textbf{Step 3: Include Irrelevant Categories}\\
\MySmall{Modify $\w^\WIS(C_i)$. \ $\tilde{f}_\ominus$ - desired fraction of irrelevant nodes.}}\\
\widearrow
\mybox{\textbf{Step 4: Tiny and Unknown Categories}\\
\MySmall{Modify $\lowhat{\vol}(C_i)$. \ $\gamma$ - maximal resolution.}}\\
\widearrow
\mybox{\textbf{Step 5: Edge Conflict Resolution}\\
\MySmall{Set the weights of inter-category edges to \eqn{\ref{eq:hybrid}}.}}\\
\widearrow
\mybox{\textbf{WRW sample}\\
\MySmall{Use transition probabilities proportional to edge weights (\Sec{sec:walks}).}}\\
\widearrow
\mybox{\textbf{Correct for the bias}\\
\MySmall{Apply formulas from \Sec{subsec:Correcting for the bias in RW, WRW and WIS}.}}\\
\widearrow\vspace{-0.04cm}
\mybox{\textbf{Final result}}\vspace{-0.6cm}
\end{center}
\caption{Overview of our approach.}
\label{fig:diagram}
\vspace{-0.3cm}
\end{figure}

\subsection{ Our practical solution: \heur}\label{subsec:heur}

\subsubsection{Step 1: Estimation of Category Volumes}
\label{subsec:Estimation of Category Volumes}
In general, we have no prior information about $G$ or $G^C$. 
Fortunately, it is easy and inexpensive 
estimate the relative category volumes $f^\sss{vol}_i$ 
which is the first piece of information we need in \eqn{\ref{eq:w_e(C_i)}} (see footnote~\ref{footnote:constant factor}).
Indeed, it is enough to run a relatively short pilot RW, and plug the collected sample~$S$ in \eqn{\ref{vol c 2 RW}} derived in Appendix B, as follows
$$
	\est{f}^\sss{vol}_i \ =\  \ \frac{1}{n}  \sum_{u\in S}\left(\frac{1}{\deg(u)}\sum_{v \in \mathcal{N}(u)}\!\! 1_{\{v\in C_i\}}\right).
$$

\subsubsection{Step 2: Category Weights Optimal Under WIS}
In order to find the optimal WIS category weights~$\w^\WIS(C_i)$ in \eqn{\ref{eq:w_e(C_i)}}, we first calculate $n_i^\opt$ as shown, under various scenarios, in Sec.~\ref{sec:Node weight setting under WIS}. 
Next, we plug the resulting~$n_i^\opt$ in~\eqn{\ref{w(C_i) propto n_i}}, \eg by setting $\w^\WIS(C_i) = n^\opt_i$.

\subsubsection{Step 3: Irrelevant Categories}
\label{subsec:Step 3: Irrelevant Categories}

\smallskip
\noindent\textbf{Problem: Potentially poor or no convergence.}
Consider the toy example in Fig.~\ref{fig:tiny_example}(a). 
We are interested in finding the relative sizes of red (dark) and green (light) categories. The white node in the middle is irrelevant for our measurement objective.
Due to symmetry, we distinguish between two types of edges with weights $w_1$ and $w_2$.
Under WIS, \eqn{\ref{eq: Relative sizes of node categories}} gives us the optimal weights $w_1>0$ and $w_2=0$, \ie WIS samples every non-white node with the same probability and never samples the white one. 
However, under WRW with these weights, relevant nodes get disconnected into two components and WRW does not converge. 
We observed a similar problem in Fig.~\ref{fig:illustration_magnifying_glass}.

\smallskip
\noindent\textbf{Guideline: Occasionally visit irrelevant nodes.}
We show in Appendix~D that the optimal WRW weights in Fig.~\ref{fig:tiny_example}(a) are $w_1=0$ and $w_2>0$.  
In that case, half of the samples are due to visits in the white (irrelevant) node. 
In other words, WRW may benefit from allocating small weight $\w(C_\ominus)\!>\!0$ to category $C_\ominus$ that groups all (if any) categories irrelevant to our estimation.
The intuition is that irrelevant nodes may not contribute to estimation but may be needed for connectivity or fast mixing.

\smallskip
\noindent\textbf{Implementation in \heur.}
In \heur, we achieve this goal by replacing $\w^\WIS(C_i)$ with
\begin{equation}
\tilde{\w}^\WIS(C_i) = \left\{
	\begin{array}{ll}
			\w^\WIS(C_i) & \textrm{if } C_i \neq C_\ominus\\			
			\tilde{f}_\ominus \cdot \sum_{C\neq C_\ominus} \w^\WIS(C) & \textrm{if } C_i = C_\ominus.		
		\end{array}	
		\right.
\end{equation}
The parameter $0\leq\tilde{f}_\ominus\ll 1$ controls the desired fraction of visits in $C_\ominus$.

\begin{figure}[t]
		\psfrag{a}[l]{$w_1$}
		\psfrag{b}[c]{$w_2$}
		\psfrag{A}[l]{\textbf{(a)}}
		\psfrag{B}[l]{\textbf{(b)}}
		\psfrag{X}[lt]{\parbox[c]{3.6cm}{\small{WIS:\hspace{0.22cm} $w_1\!>\!0,\ w_2\eq0$\\WRW: $w_1\eq 0,\ w_2\!>\!0$}}}  %
		\psfrag{Y}[lt]{\parbox[c]{4.6cm}{\small{WIS:\hspace{0.23cm} $w_1 = 190\, w_2$\\
		WRW: $w_1 \cong 60\, w_2$ \textrm{ for }n=50\\
		${}$\qquad\quad \ $w_1 \cong 100\, w_2$ \textrm{ for }n=500\\
		${}$\qquad\quad \ $w_1 = 190\, w_2$ \textrm{ for }$n\rightarrow\infty$}}}  %
    \includegraphics[width=0.46\textwidth]{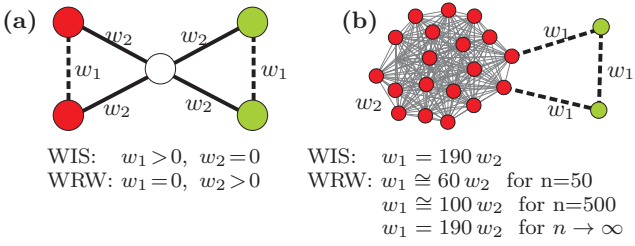}
\caption{Optimal edge weights: WIS vs WRW. The objective is to
compare the sizes of red (dark) and green (light) categories. 
}
\label{fig:tiny_example}
\vspace{-10pt}
\end{figure}

\subsubsection{Step 4: Tiny and Unknown Categories}
\label{subsec:Tiny and Unknown Categories}

\smallskip
\noindent\textbf{Problem: ``black holes''.}
Every optical system has a fundamental magnification limit due to diffraction 
and our ``graph magnifying glass'' is no exception. Consider the toy graph in Fig.~\ref{fig:tiny_example}(b): it consists of a big clique $C_\sss{big}$ of 20 red nodes with edge weights~$w_2$, and a green category~$C_\sss{tiny}$ with two nodes only and edge weights $w_1$. In \Sec{subsec: Relative sizes of node categories}, we saw that WIS optimally estimates the relative sizes of red and green categories for $\w(C_\sss{big})\eq\w(C_\sss{tiny})$, \ie for $w_1\eq 190\,w_2$. 
However, for such large values of $w_1$, the two green nodes behave as a sink (or a ``black hole'') for a WRW of finite length, thus  increasing the variance of the category size estimation. 
 
\smallskip
\noindent\textbf{Guideline: limit edge weights.}
In other words, although WIS suggests to over-sample small categories, WRW should ``under-over-sample'' very small categories to avoid black holes. 
For example, in Fig.~\ref{fig:tiny_example}(b) $w_1\simeq 60\,w_2 \ (\ll 190 w_2)$ is optimal for WRW of length~$n\eq 50$ (simulation results).

\smallskip
\noindent\textbf{Implementation in \heur.}
In \heur, we achieve this goal by replacing $\vol(C_i)$ in \eqn{\ref{eq:w_e(C_i)}} with
\begin{eqnarray}
\label{eq:volC max} \lowtilde{\vol}(C) &=& \max\Big\{\lowhat{\vol}(C),\ \vol_{min}\Big\},\quad \textrm{ where }\\
\label{eq:vol_min} \vol_{min} &=& \frac{1}{\gamma}\cdot \max_{C\neq C_\ominus}\{\lowhat{\vol}(C)\}.
\end{eqnarray}
Moreover, this formulation takes care of every category~$C$ that was not discovered by the pilot RW in \Sec{subsec:Estimation of Category Volumes}, by setting $\lowtilde{\vol}(C)\eq \vol_{min}$.


\subsubsection{Step 5: Edge Conflict Resolution}
\label{subsec:Step 5: Edge Conflict Resolution}

\smallskip
\noindent\textbf{Problem: Conflicting desired edge weights.} 
With the above modifications, our target edge weights defined in Eq.(\ref{eq:w_e(C_i)}) can be rewritten as
\begin{equation}\label{eq:tilde w_e(C_i)}
	\tilde{\w}_e(C_i)\ =\ \frac{\tilde{\w}^\WIS(C_i)}{\lowtilde{\vol}(C_i)}.
\end{equation}
We can directly set the weight $\w(u,v)\eq\tilde{\w}_e(C(u))\eq\tilde{\w}_e(C(v))$ 
for every intra-category edge $\{u,v\}$. 
However, for every inter-category edge, we usually have ``conflicting'' weights  $\tilde{\w}_e(C(u))\neq \tilde{\w}_e(C(v))$ desired at the two ends of the edge.

\smallskip
\noindent\textbf{Guideline: prefer inter-category edges.}
There are several possible edge weight assignments that achieve the desired category node weights. High weights on intra-category edges and small weights on inter-category edges result in WRW staying in small categories~$C_\sss{tiny}$ for a long time. 
In order to improve the mixing time, we should do exactly the opposite, \ie  assign relatively high weights to inter-category edges (connecting relevant categories). As a result, WRW will enter $C_\sss{tiny}$ more often, but will stay there for a short time.
This intuition is motivated by Monte Carlo variance reduction techniques such as the use of \emph{antithetic variates}~\cite{Gentle2003}, which seek to induce negative correlation between consecutive draws so as to reduce the variance of the resulting estimator.

\smallskip
\noindent\textbf{Implementation in \heur.} We choose to assign an edge weight  $\tilde{\w}_e$ that is in between these two values $\tilde{\w}_e(C(u))$ and $\tilde{\w}_e(C(v))$.
We considered several candidate  such assignments.
We may take the {\em arithmetic} or {\em geometric} mean of the conflicting weights, which we denote by $\w^\sss{ar}(u,v)$ and $\w^\sss{ge}(u,v)$, respectively. 
We may also use the {\em maximum} of the two values, $\w^\sss{max}(u,v)$, which should improve mixing according to the discussion above.
However, $\w^\sss{max}(u,v)$ alone would also add high weight to irrelevant nodes~$C_\ominus$ (possibly far beyond $\tilde{f}_\ominus$). 
To avoid this undesired effect, we distinguish between the two cases by defining a hybrid solution:
\begin{equation}\label{eq:hybrid}
	\w^\sss{hy}(u,v) = \left\{
        \begin{array}{ll}
                        \w^\sss{ge}(u,v)  & \textrm{if }C_\ominus\in\{C(u),C(v)\}\\
                        \w^\sss{max}(u,v) &\textrm{otherwise.}
                \end{array}
                \right.	
\end{equation}
This hybrid edge assignment was the one we found to work best in practice - see Section \ref{Applications}.

\subsection{Discussion}


\subsubsection{Information needed about the neighbors}
\label{subsec:Information needed about the neighbors}
In the pilot RW (\Sec{subsec:Estimation of Category Volumes}) as well as in the main WRW, 
we assume that by sampling a node~$v$ we also learn the category (but not degree) of each of its neighbors $u\in\mathcal{N}(v)$. 
Fortunately, such information is often available in most online graphs at no additional cost, especially when scraping html pages (as we do). 
For example, when sampling colleges in Facebook (\Sec{Applications}), we use the college membership information of all~$v$'s neighbors, which, in Facebook, is available at~$v$ together with the friends list.


\subsubsection{Cost of pilot RW}
\label{subsec:Cost of pilot RW}
The pilot RW volume estimator described in~\Sec{subsec:Estimation of Category Volumes} considers the categories not only of the sampled nodes, but also of their neighbors.
As a result, it achieves high efficiency, as we show in simulations (\Sec{subsec:Estimating volumes is usually cheap}) and Facebook measurements (\Sec{subsec:Measurement Setup}). 
Given that, and high robustness of \heur to estimation errors (see~\Sec{Robustness to gamma and volume estimation}), pilot RW should be only a small fraction of the later WRW (\eg 6.5\% in our Facebook measurements in \Sec{Applications}).


\subsubsection{Setting the parameters}
\label{subsec:Setting the parameters}

\heur sets the edge weights trying to achieve roughly $\w^\WIS(C_i)$ as the main goal. We slightly shape $\w^\WIS(C_i)$ to avoid black holes and improve mixing, which is controlled by two natural and easy-to-interpret parameters, $\tilde{f}_\ominus$ and $\gamma$.

\smallskip
\noindent\textbf{Irrelevant nodes visits~$\tilde{f}_\ominus$.}\quad
The parameter $0\leq\tilde{f}_\ominus\ll 1$ controls the desired fraction of visits in $C_\ominus$.
When setting~$\tilde{f}_\ominus$, we should exploit the information provided by the pilot crawl. 
If the relevant categories appear poorly interconnected and often separated by irrelevant nodes, we should set $\tilde{f}_\ominus$ relatively high. 
We have seen an extreme case in \Fig{fig:tiny_example}(a), with disconnected relevant categories and optimal~$\tilde{f}_\ominus\eq0.5$.
In contrast, when the relevant categories are strongly interconnected, we should use much smaller~$\tilde{f}_\ominus$. 
However, because we can never be sure that the graph induced on relevant nodes is connected, we recommend always using~$\tilde{f}_\ominus>0$. 
For example, when measuring Facebook in \Sec{Applications}, we set~$\tilde{f}_\ominus=1\%$. 

\smallskip
\noindent\textbf{Maximal resolution~$\gamma$.}\quad
The parameter $\gamma\geq 1$ can be interpreted as the maximal resolution of our ``graph magnifying glass'', with respect to the largest relevant category~$C_\sss{big}$. 
\heur will typically sample well all categories that are less than $\gamma$~times smaller than $C_\sss{big}$; all categories smaller than that are relatively undersampled (see~\Sec{subsec:Robustness to the choice of gamma}).
In the extreme case, for $\gamma\rightarrow\infty$, \heur tries to cover every category, no matter how small, which may cause the ``black hole'' problem discussed in~\Sec{subsec:Tiny and Unknown Categories}. 
In the other extreme, for $\gamma\eq1$ (and identical~$\w^\WIS(C_i)$ for all categories, including $C_\ominus$), 
\heur reduces to RW. \quad 
We recommend always setting~$1<\gamma<\infty$. Ideally, we know $|C_\sss{smallest}|$ - the smallest category size that is still relevant to us. In that case we should set~$\gamma=|C_\sss{big}|/|C_\sss{smallest}|$.\footnote{Strictly speaking, $\gamma$ is related to volumes~$\vol(C_i)$ rather than sizes $|C_i|$. They are equivalent when category volume is proportional to its size, which is often the case, and is the central assumption in the ``scale-up method''~\cite{Bernard2010}.}
For example, in \Sec{Applications} the categories are US colleges; we set $\gamma\eq1000$, because colleges with size smaller than 1/1000th of the largest one (\ie with a few tens of students) seem irrelevant to our measurement. \quad 
As another rule of thumb, we should try to set smaller~$\gamma$ for relatively small sample sizes and in graphs with tight community structure (see~\Sec{Robustness to gamma and volume estimation}).

\subsubsection{Conservative approach} 

Note that a reasonable setting of these parameters (\ie $\tilde{f}_\ominus>0$ and $1<\gamma<\infty$, and any conflict resolution discussed in the paper), increases the weights of large categories (including $C_\ominus$) and decreases the weight of small categories, compared to $\w^\WIS(C_i)$.
This makes \heur allocate category weights between the two extremes: RW 
and WIS. 
Consequently, \heur can be considered \emph{conservative} (with respect to WIS). 

\subsubsection{\heur is unbiased}
It is also important to note that because the collected WRW sample is eventually corrected with the actual sampling weights as described in~\Sec{subsec:Correcting for the bias in RW, WRW and WIS}, \heur estimation process is \emph{unbiased}, regardless of the choice of weights (so long as convergence is attained). 
In contrast, suboptimal weights (\eg due to estimation error of~$\est{f}^\sss{vol}_C$) can increase WRW mixing time, and/or the \emph{variance} of the resulting estimator. 
However, our simulations and empirical experiments on Facebook (see Sec.~5 and~6) show that \heur is very robust to suboptimal choice of weights.

\begin{figure*}
\begin{center}
		\psfrag{A}[r][]{\textbf{(a)}}
		\psfrag{B}[r][]{\textbf{(b)}}
		\psfrag{C}[r][]{\textbf{(c)}}
		\psfrag{D}[r][]{\textbf{(d)}}
		\psfrag{E}[r][]{\textbf{(e)}}
		\psfrag{F}[r][]{\textbf{(f)}}
		\psfrag{G}[r][]{\textbf{(g)}}
		\psfrag{H}[r][]{\textbf{(h)}}
		\psfrag{I}[r][]{\textbf{(i)}}
		\psfrag{J}[r][]{\textbf{(j)}}
		\psfrag{K}[r][]{\textbf{(k)}}
		\psfrag{L}[r][]{\textbf{(l)}}
		\psfrag{M}[r][]{\textbf{(m)}}
		\psfrag{N}[r][]{\textbf{(n)}}		
		\psfrag{g5}[c][c][0.7]{$\gamma=5$}
		\psfrag{opt}[c][c][0.7]{optimal}
    \psfrag{error}[c][t]{$\NRMSE(\est{f}_\sss{tiny})$}
    \psfrag{volumeError}[c][]{$\NRMSE(\est{f}^\sss{vol}_\sss{tiny})$}
    \psfrag{fraction}[c][][1.0]{$\Prob[C_\sss{tiny}\textrm{ visited}]$}
    \psfrag{gain alpha}[c][]{gain $\alpha$}
    \psfrag{weight w}[c][][0.9]{weight $w$ (equivalent to $\gamma$)}
    \psfrag{sample length}[c][]{sample length $n$}
    \psfrag{WRW sample length}[c][]{WRW sample length $n$}
\includegraphics[width=0.98\textwidth]{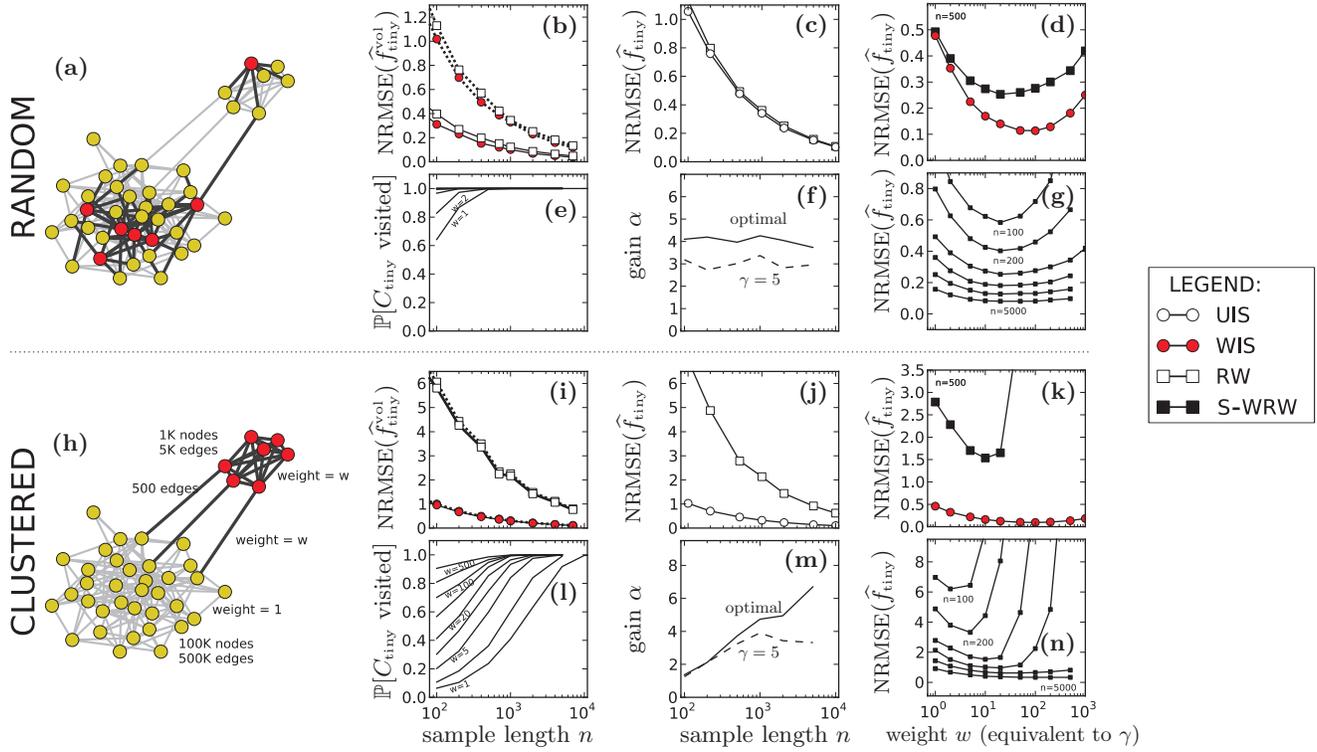}
\vspace{-5pt}
\caption{ RW and \heur under two scenarios: Random \textbf{(a-g)} and Clustered (h-n). 
In (b,i), we show error of two volume estimators: naive \eqn{\ref{vol c 1 RW}} (dotted) and neighbor-based \eqn{\ref{vol c 2 RW}} (plain). 
Next, we show error of size estimator as a function of $n$~(c,j) and $w$~(d,g,k,n); in the latter, UIS and RW correspond to WIS and \heur for $w\eq 1$. 
In (e,l), we show the empirical probability that \heur visits $C_\sss{tiny}$ at least once.
Finally, (f,m) is gain~$\alpha$ of \heur over RW under the optimal choice of $w$ (plain), and for fixed $\gamma\eq w\eq5$ (dashed). 
}
\label{fig:simulation_results}
\vspace{-19pt}
\end{center}
\end{figure*}

\section{Simulation results}\label{Simulation results}

The gain of our approach compared to RW comes from two main factors.
First, \heur avoids, to a large extent or completely, the nodes in~$C_\ominus$ that are irrelevant to our measurement.
This fact alone can bring an arbitrarily large improvement ($\frac{N}{N-|C_\ominus|}$ under WIS), especially when $C_\ominus$ is large compared to~$N$.
We demonstrate this in the Facebook measurements in Section~\ref{Applications}.
Second, we can better allocate samples among the relevant categories. This factor is observable in our Facebook measurements as well, but it is more difficult to evaluate due to the lack of ground-truth therein.
In this section, we evaluate the optimal allocation gain in a controlled simulation and we demonstrate some key insights.

\subsection{Setup}
We consider a graph~$G$ with 101K nodes and 505.5K edges organized in two densely (and randomly) connected communities\footnote{The term ``community'' refers to cluster and is defined purely based on topology. The term ``category'' is a property of a node and is independent of topology.} as shown in Fig.~\ref{fig:simulation_results}(h).

The nodes in~$G$ are partitioned into two node categories: $C_\sss{tiny}$ with 1K  nodes (dark red), and $C_\sss{big}$ with 100K nodes (light yellow). We consider two extreme scenarios of such a partition. The `random' scenario is purely random, as shown in Fig.~\ref{fig:simulation_results}(a). 
In contrast, under `clustered', categories $C_\sss{tiny}$ and $C_\sss{big}$ coincide with the existing communities in $G$, as shown in Fig.~\ref{fig:simulation_results}(h). It is arguably the worst case scenario for graph sampling by exploration.


We fix the edge weights of all internal edges in $C_\sss{big}$ to~1. All the remaining edges, \ie all edges incident on nodes in category~$C_\sss{tiny}$, have weight $w$ each, where $w\geq 1$ is a parameter. Note that this is equivalent to setting $\tilde{\w}_e(C_\sss{big})\eq1$, $\tilde{\w}_e(C_\sss{tiny})\eq w$, and `max' or `hybrid' conflict resolution. 


\subsection{Measurement objective and error metric}
We are mainly interested in measuring the relative sizes $f_\sss{tiny}$ and $f_\sss{big}$ of categories~$C_\sss{tiny}$ and~$C_\sss{big}$, respectively. 

We use Normalized Root Mean Square Error (\NRMSE) to assess the estimation error, defined as~\cite{Ribeiro2010}:
\begin{equation}\label{eq:NRMSE}
\NRMSE(\est{x}) = \frac{ \sqrt{\Mean\big[ (\est{x}-x)^2 \big] }}{x},
\end{equation}
where~$x$ is the real value and $\est{x}$ is the estimated one.

\subsection{Results}

\subsubsection{Estimating volumes is usually cheap}
\label{subsec:Estimating volumes is usually cheap}
The first step in \heur is obtaining category volume estimates $\est{f}^\sss{vol}_i$. 
We achieve it by running a short pilot RW and applying the estimator~\eqn{\ref{vol c 2 RW}}. We show $\NRMSE(\est{f}^\sss{vol}_\sss{tiny})$ as plain curves in~\Fig{fig:simulation_results}(b). 
This estimator takes advantage of the knowledge of the categories of the neighboring nodes, which makes it much more efficient than the naive estimator~\eqn{\ref{vol c 1 RW}} shown by dashed curves.
Moreover, the advantage of \eqn{\ref{vol c 2 RW}} over \eqn{\ref{vol c 1 RW}} grows with the graph density and the skewness of its degree distribution (not shown here). 

Note that under `random', RW and WIS (with the sampling probabilities of RW) are almost equally efficient.  
However, on the other extreme, \ie under the `clustered' scenario, the performance of RW becomes much worse and the advantage of \eqn{\ref{vol c 2 RW}} over \eqn{\ref{vol c 1 RW}} diminishes. 
This is because essentially all friends of a node from category~$C_i$ are in $C_i$ too, which reduces formula \eqn{\ref{vol c 2 RW}} to \eqn{\ref{vol c 1 RW}}.
Nevertheless, we show later in~\Sec{Robustness to gamma and volume estimation} that even severalfold volume estimation errors are likely not to affect significantly the results.

\subsubsection{Visiting the tiny category}
Fig.~\ref{fig:simulation_results}(e,l) presents the empirical probability $\Prob[C_\sss{tiny}\textrm{ visited}]$ that our walk visits at least one node from~$C_\sss{tiny}$. 
Of course, this probability grows with the sample length. 
However, the choice of weight~$w$ also helps in it. Indeed, WRW with $w>1$ is more likely to visit~$C_\sss{tiny}$  than RW ($w=1$, bottom line). 
This demonstrates the first advantage of introducing edge weights and WRW.

\subsubsection{Optimal $w$ and $\gamma$}

Let us now focus on the estimation error as a function of~$w$, shown in Fig.~\ref{fig:simulation_results}(d,k).
Interestingly, this error does not drop monotonically with~$w$ but follows a 'U' shaped function with a clear optimal value~$w^\opt$.

Under WIS, we have $w^\opt\simeq 100$, which confirms our findings in \Sec{subsec: Relative sizes of node categories}.
Indeed, according to \eqn{\ref{eq: Relative sizes of node categories}}, we need the same number of samples 
from the two categories, and thus $\w^\WIS(C_\sss{tiny})=\w^\WIS(C_\sss{big})$ (by~\eqn{\ref{w(C_i) propto n_i}}). 
By plugging this and $\vol(C_\sss{big})=100\cdot\vol(C_\sss{tiny})$ to \eqn{\ref{eq:w_e(C_i)}}, we finally obtain the WIS-optimal edge weights in $C_\sss{tiny}$, \ie
$w^\opt=\w_e(C_\sss{tiny}) = 100\cdot\w_e(C_\sss{big})=100$.\footnote{For simplicity, we ignored in this calculation the conflicts on the 500 edges between $C_\sss{big}$ and $C_\sss{tiny}$.}

In contrast, WRW is optimized for $w<100$. For the sample length $n\eq 500$ as in Fig.~\ref{fig:simulation_results}(d,k), the error is minimized already for $w^\opt\!\simeq\!20$ and increases for higher weights. 
This demonstrates the ``black hole'' effect discussed in~\Sec{subsec:Tiny and Unknown Categories}. 
It is much more pronounced in the `clustered' scenario, confirming our intuition that black-holes become a problem only in the presence of relatively isolated, tight communities.
Of course, the black hole effect diminishes with the sample length~$n$ (and completely vanishes for~$n\!\rightarrow\!\infty$), which can be observed in Fig.~\ref{fig:simulation_results}(g,n), especially in~(n).

In other words, the optimal assignment of edge weights (in relevant categories) under WRW lies somewhere between RW (all weights equal) and  WIS. 
In \heur, we control it by parameter~$\gamma$. 
In this example, we have $\gamma\equiv w$ for $\gamma\le 100$. Indeed, by combining \eqn{\ref{eq:w_e(C_i)}}, \eqn{\ref{eq:volC max}}, \eqn{\ref{eq:vol_min}}, $\w^\WIS(C_\sss{tiny})\eq\w^\WIS(C_\sss{big})$, we obtain
\begin{eqnarray}
\nonumber w &=& \frac{w}{1} \ =\  \frac{w_e(C_\sss{tiny})}{w_e(C_\sss{big})}\ =\ 
 \frac{\w^\WIS(C_\sss{tiny})/\lowtilde{\vol}(C_\sss{tiny})}{\w^\WIS(C_\sss{big})/\lowtilde{\vol}(C_\sss{big})}\\  
\nonumber &=& \frac{\lowtilde{\vol}(C_\sss{big})}{\lowtilde{\vol}(C_\sss{tiny})}\ 
  =\ \frac{\vol(C_\sss{big})}{\frac{1}{\gamma}\vol(C_\sss{big})}\ =\ \gamma.	
\end{eqnarray}
Consequently, the optimal setting of $\gamma$ is the same as $w^\opt$ discussed above.


\subsubsection{Gain $\alpha$}

The gain $\alpha$ of WIS over UIS is given by~\eqn{\ref{eq: Relative sizes of node categories}}.
In this case, we have $\alpha=(101K)^2\cdot (4\cdot1K\cdot 100K)^{-1} \simeq 25$.
Indeed, WIS with $n\eq500$ samples shown in Fig.~\ref{fig:simulation_results}(d) achieves $\NRMSE\!\simeq\!0.1$, which is the same as UIS of about $\alpha\eq25$ times more samples (see Fig.~\ref{fig:simulation_results}(c)).

This gain due to stratification is smaller for sampling by exploration: a 500-hop-long WRW with $w\!\simeq\!20$ yields the same error $\NRMSE\!\simeq\!0.3$ as a 2000-hop-long RW.
This means that WRW reduces the sampling cost by a factor of~$\alpha\simeq 4$. 
Fig.~\ref{fig:simulation_results}(f) shows that this gain does not vary much with the sampling length.
Under `clustered', both RW and WRW perform much worse. 
Nevertheless, Fig.~\ref{fig:simulation_results}(m) shows that also in this scenario WRW may significantly reduce the sampling cost, especially for longer samples.


It is worth noting that WRW can sometimes significantly outperform UIS. This is the case in  Fig.~\ref{fig:simulation_results}(d), where UIS is equivalent to WIS with $w\eq1$.
Because no walk can mix faster than UIS (that is independent and thus has perfect mixing), improving the mixing time alone~\cite{Boyd2004_mixing,Ribeiro2010,Ribeiro2010a,Avrachenkov2010} cannot achieve the potential gains of stratification, in general.

So far we focused on the smaller set $C_\sss{tiny}$ only.
When estimating the size of $C_\sss{big}$, all errors are much smaller, but we observe similar gain~$\alpha$.

\subsubsection{Robustness to $\gamma$ and volume estimation}
\label{Robustness to gamma and volume estimation}
The gain $\alpha$ shown above is calculated for the optimal choice of $w$, or, equivalently, $\gamma$. Of course, in practice it might be impossible to obtain this value. Fortunately, \heur is relatively robust to the choice of parameters. 
The dashed lines in Fig.~\ref{fig:simulation_results}(f,m) are calculated for~$\gamma$ fixed to~$\gamma\eq 5$, rather than optimized. 
Note that this value is often drastically smaller than the optimal one (\eg $w^\opt\!\simeq\!50$ for $n\eq5000$). 
Nevertheless, although the performance somewhat drops, \heur still reduces the sampling cost about three-fold.

This observation also addresses potential concerns one might have regarding the category volume estimation error (see \Sec{subsec:Estimation of Category Volumes}). 
Indeed, setting $\gamma\eq 5$ means that every category~$C_i$ with volume estimated at $\lowhat{\vol}(C_i)\le\frac{1}{5}\vol(C_\sss{big})$ is treated the same. 
In Fig.~\ref{fig:simulation_results}(f), the volume of $C_\sss{tiny}$ would have to be overestimated by more than 20 times in order to affect the edge weight setting and thus the results. We have seen in~\Sec{subsec:Estimating volumes is usually cheap} that this is very unlikely, even under smallest sample lengths and most adversarial scenarios. 




\subsection{Summary}

WRW brings two types of benefits (i) avoiding irrelevant nodes $C_\ominus$ and (ii) carefully allocating samples between relevant categories of different sizes. Even when $C_\ominus\eq\emptyset$, WRW can still reduce the sampling cost by 75\%.
This second benefit is more difficult to achieve when the categories form strong and tight communities, which leads to the ``black hole''' effect. We should then choose smaller, more conservative values of $\gamma$ in \heur, which translate into smaller $w$ in our example.
In contrast, under a looser community structure this problem disappears and WRW is closer to WIS.

\section{Implementation in Facebook}
\label{Applications}

\begin{figure*}[!ht]
    \psfrag{ra}[l][][0.5]{RW}
    \psfrag{ge}[l][][0.5]{\heur, geometric}
    \psfrag{ar}[l][][0.5]{\heur, arithmetic}
    \psfrag{hy}[l][][0.5]{\heur, hybrid}
    \psfrag{ra2}[l][][0.5]{RW}
    \psfrag{ge2}[l][][0.5]{geometric}
    \psfrag{ar2}[l][][0.5]{arithmetic}
    \psfrag{hy2}[l][][0.5]{hybrid}
    
    \psfrag{SZ}[c][][0.7]{Relative size $\est{f}_i$}
    \psfrag{NT}[c][][0.7]{Number of samples $n_i$}
    \psfrag{NRMSE}[c][][0.7]{Average $\NRMSE(\est{f}_i)$}
    \psfrag{RK}[c][][0.7]{Ranked colleges}
    \psfrag{NW}[c][][0.7]{Number of samples~$n$}
    \psfrag{ten}[c][][0.55]{$\bf{\times} 10^{-6}$}
    
    \psfrag{A}[c][]{\textbf{(a)}}
		\psfrag{B}[c][]{\textbf{(b)}}
		\psfrag{C}[c][]{\textbf{(c)}}
		\psfrag{D}[c][]{\textbf{(d)}}
		\psfrag{E}[c][]{\textbf{(e)}}
		\psfrag{F}[c][]{\textbf{(f)}}
		
    \includegraphics[width=1.\textwidth]{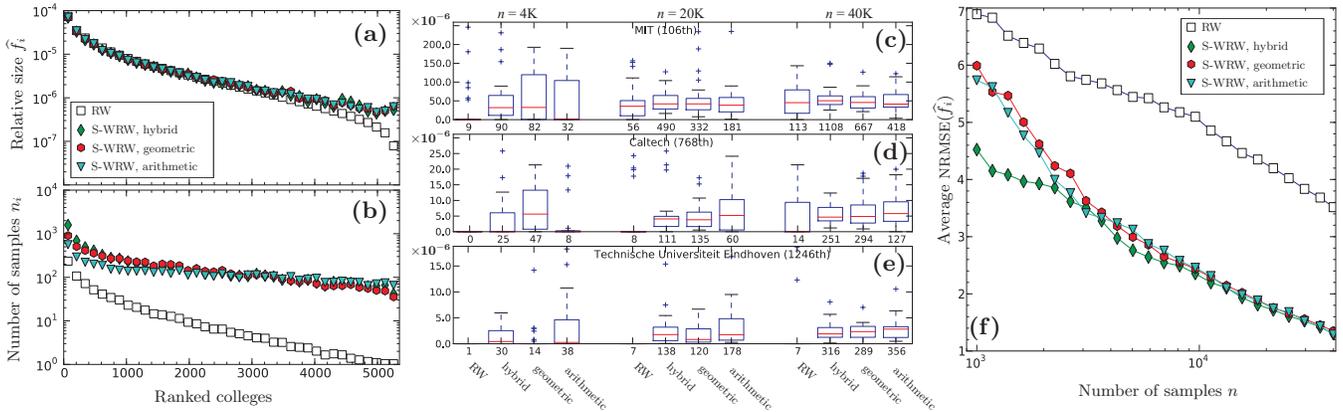}
    \vspace{-0.8cm}
\caption{
5331 colleges discovered and ranked by RW. 
(a)~Estimated relative college sizes~$\est{f}_i$. (b)~Absolute number of user samples per college.
(c-e)~25 estimates of size~$\est{f}_i$ for three different colleges and sample lengths~$n$.  
(f)~Average \NRMSE \ of college size estimation.
\quad Results in (a,b,f) are binned. 
}
\label{fig:combo_Facebook}
\vspace{-0.2cm}
\end{figure*}


As a concrete application, we apply \heur to measure the Facebook social graph, which is our motivating and canonical example. We also note that it is an undirected and can also be considered a static graph, for all practical purposes in this study.\footnote{The Facebook characteristics do change but in time scales much longer than the 3-day duration of our crawls. Websites such as Facebook statistics, Alexa etc show that the number of Facebook users is growing with rate 0.1-0.2\% per day.}
In Facebook, every user may declare herself a member of a college\footnote{There also exist categories other than colleges, namely ``work'' and ``high school''. Facebook
requires a valid category-specific email for verification.}
he/she attends. This membership information is publicly available by default and allows us to answer some interesting questions.
For example, how do the college networks (or ``colleges'' for short) compare with respect to their sizes? 
What is the college-to-college friendship graph?
In order to answer these questions, we have to collect many college user samples, preferably evenly distributed between colleges. This is the main goal of this section.


\subsection{Measurement Setup}
\label{subsec:Measurement Setup}

By default, every Facebook user can see the basic information on any other user, including the name, photo, and a list of friends together with their college memberships (if any).
We developed a high performance multi-threaded crawler to explore Facebook's social graph by scraping this web interface.

To  make informed decision for the parameters of \heur, we first ran a short pilot RW (see \Sec{subsec:Estimation of Category Volumes}) with a total of~$65K$ samples (which is only 6.5\% of the length of the main \heur sample). 
Although our pilot walk visited only 2000 colleges, it estimated the relative volumes~$f^\vol_i$ for about 9500 colleges discovered among friends of sampled users, as discussed in~\Sec{subsec:Cost of pilot RW}.
%
In~\Fig{fig:pilot RW}(a), we show that the neighbor-based estimator~\eqn{\ref{vol c 2 RW}} greatly outperforms the naive estimator~\eqn{\ref{vol c 1 RW}}. 
These volumes cover several decades. 
Because colleges with only a few tens of users are not of our interest, we set the maximal resolution to~$\gamma\eq 1000$ (see the discussion in~\Sec{subsec:Setting the parameters}). 
Finally, because the college students looked very well interconnected in our pilot RW, we set the desired fraction of irrelevant nodes to a small number~$\tilde{f}_\ominus\eq1\%$. 


In the main measurement phase, we collected three \heur crawls, each with different edge weight conflict resolution (hybrid, geometric, and arithmetic), and one simple RW crawl as a baseline comparison~(Table \ref{tab:fb_datasets}). 
For each crawl type we collected 1~million \emph{unique} users. 
Some of them are sampled multiple times (at no additional bandwidth cost), which results in higher total number of samples in the second row of~Table \ref{tab:fb_datasets}. 
Our crawls were performed on Oct. 16-19 2010, and are available at~\cite{wrw-dataset}.

\subsection{Results: RW vs. \heur}
\label{subsec:FB results}

\begin{table}
\flushleft
\small
\begin{tabular}{@{}l|c|c|c|c@{}}
 							                &  RW        &   \multicolumn{3}{c}{\heur} \\ \cline{3-5} 
 							                &          &   Hybrid      & Geometric    & Arithmetic  \\
\hline 
Unique samples                       &  1,000K     &   1,000K      &   1,000K     &    1,000K  \\
Total samples                         &  1,016K      &   1,263K     &    1,228K     &    1,237K  \\
College samples                      &  9\%  &   86\% &  79\% &    58\%    \\
Unique Colleges                    &  5,331      &   9,014       &    8,994      &   10,439
\end{tabular}
\caption{Overview of collected Facebook datasets.
}
\label{tab:fb_datasets}
\vspace{-5pt}
\end{table}

\subsubsection{Avoiding irrelevant categories}

Only 9\% of the RW's samples come from colleges, which means that the vast majority of sampling effort is wasted.
In contrast, the \heur crawls achieved 6-10 better efficiency, collecting 86\% (hybrid), 79\% (geometric) and 58\% (arithmetic) samples from colleges.
Note that these values are significantly lower than the target 99\% suggested by our choice of $\tilde{f}_\ominus\eq1\%$, and that \heur hybrid reaches the highest number. This is in agreement with our discussion in \Sec{subsec:Step 5: Edge Conflict Resolution}.
Finally, we also note that \heur crawls discovered $1.6-1.9$ times more unique colleges than RW.

It might seem surprising that RW samples colleges in 9\% of cases while only 3.5\% of Facebook users belong to colleges. 
This can be explained by looking at the last rows of Table~\ref{tab:fb_datasets}. 
Indeed, the college users have on average three times more Facebook friends than average users, 
and therefore they attract RW approximately three times more often.

\subsubsection{Stratification}

The advantage of \heur over RW does not lie exclusively in avoiding the nodes in the irrelevant category~$C_\ominus$. \heur can also over-sample small categories (here colleges) at the cost of under-sampling large ones (which are very well sampled anyway). This feature becomes important especially when the category sizes differ significantly, which is the case in Facebook. Indeed, Fig.~\ref{fig:combo_Facebook}(a) shows that college sizes exhibit great heterogeneity.
For a fair comparison, we only include the 5,331 colleges discovered by RW. (In fact, this filtering actually gives preference to RW. 
\heur crawls discovered many more colleges that we do not show in this figure.)
They span more than two orders of magnitude and follow a heavily skewed distribution (not shown here).

Fig.~\ref{fig:combo_Facebook}(b) confirms that \heur successfully oversamples the small colleges. Indeed, the number of \heur samples per college is almost constant (roughly around 100). In contrast, the number of RW samples follows closely the college size, which results in dramatic 100-fold differences between RW and \heur for smaller colleges.



\subsubsection{College size estimation}
With more samples per college, we naturally expect a better estimation accuracy under \heur. We demonstrate it for three colleges of different sizes (in terms of the number of Facebook users): MIT (large), Caltech (medium), and  Eindhoven University of Technology (small).
Each boxplot in Fig.~\ref{fig:combo_Facebook}(c-e) is generated based on 25 independent college size estimates~$\est{f}_i$ that come from 
walks of length $n\eq4$K (left), 20K (middle), and 40K (right) samples each. 
%
For the three studied colleges, RW fails to produce reliable estimates in all cases except for MIT (largest college) under the two longest crawls. Similar results hold for the overwhelming majority of middle-sized and small colleges. The underlying reason is the very small number of samples collected by RW in these colleges, averaging at below 1 sample per walk. In contrast, the three \heur crawls contain typically 5-50 times more samples than RW (in agreement with Fig.~\ref{fig:combo_Facebook}(b)), and produce much more reliable estimates.

Finally, we aggregate the results over all colleges and compute the gain~$\alpha$ of \heur over RW. 
We calculate the error $\NRMSE(\est{f}_i)$ by taking as our ``ground truth''~$f_i$ the grand average of $\est{f}_i$ values over all samples collected via all full-length walks and crawl types.
Fig.~\ref{fig:combo_Facebook}(f) presents $\NRMSE(\est{f}_i)$ averaged over all 5,331 colleges discovered by RW, as a function of walk length~$n$. 
As expected, for all crawl types the error decreases with~$n$. 
However, there is a consistent large gap between RW and all three versions of \heur.
RW needs 13-15 times more samples than \heur in order to achieve the same error.


\begin{figure}[t]
\psfrag{ra}[l][][0.5]{$RW_pilot$}
    \psfrag{rwrwrwrwrwrwrw}[c][][0.6]{\hspace{-1.4cm}pilot RW}
    \psfrag{huhuhuhuhuhuhu}[c][][0.6]{\hspace{0.0cm}\heur, $\gamma=100$ }
    \psfrag{ththththththth}[c][][0.6]{\hspace{0.2cm}\heur, $\gamma=1000$}
    \psfrag{starstarst}[c][][0.6]{\hspace{-0.4cm}neighbor}
    \psfrag{inducedind}[c][][0.6]{\hspace{-1.0cm}naive}
    \psfrag{size}[c][b][0.75]{Relative size $\est{f}_i$}
    \psfrag{Volume}[c][b][0.75]{Relative volume $\est{f}^\vol_i$}
    \psfrag{num}[c][t][0.75]{Number of samples $n_i$}
    \psfrag{NRMSE}[c][t][0.75]{$\NRMSE(\est{f}^\vol_i)$}
\centering
\includegraphics[width=0.5\textwidth]{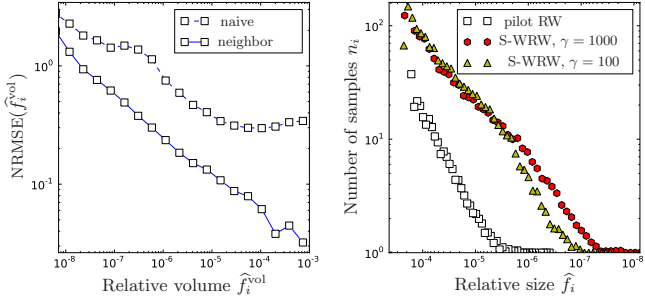}
\caption{Facebook: Pilot RW and other walks of the same length~$n\eq65K$. 
\ \textbf{(a)}~The performance of the neighbor-based volume estimator~\eqn{\ref{vol c 2 RW}} (plain line) and the naive one~\eqn{\ref{vol c 1 RW}} (dashed line). As `ground-truth' we used~$f^\vol_i$ calculated for all 4$\times$1M collected samples. 
\ \textbf{(b)}~The effect of the choice of $\gamma$.} 
\label{fig:pilot RW}
\vspace{-15pt}
\end{figure}

\subsubsection{The effect of the choice of $\gamma$}
\label{subsec:Robustness to the choice of gamma}

Recall that in all the \heur results described above, we used the resolution~$\gamma\eq 1000$. 
In order to check how sensitive the results are to the choice of this parameter, we also tried a (shorter) \heur run with $\gamma\eq100$, \ie ten times smaller. In~\Fig{fig:pilot RW}(b), we see that the number of samples collected in the smallest colleges is smaller under $\gamma\eq100$ than under $\gamma\eq1000$. In fact, the two curves diverge for colleges about 100 times smaller than the biggest college, \ie exactly at the maximal resolution $\gamma\eq100$. 

In any case, both settings of $\gamma$ perform orders of magnitude better than RW of the same length.

\subsection{Summary}
Only about 3.5\% of 500M Facebook users are college members.
There are more than 10K colleges and they greatly vary in size, ranging from 50 (or fewer) to 50K members (we aggregate students, alumni and staff). In this setting, state-of-the-art sampling methods such as RW are bound to perform poorly. Indeed, UIS, \ie an idealized version of RW, with as many as 1M samples will collect only one sample from size-500 college, on average. Even if we could magically sample directly only from colleges, we would typically collect fewer than 30 samples per size-500 college.

\heur solves these problems. 
We showed that \heur of the same length collects typically about 100 samples per size-500 college. As a result, \heur outperforms RW by $\alpha = 13-15$ times or $\alpha = 12-14$ times if we also consider the 6.5\% overhead from the initial pilot RW.
This huge gain can be decomposed into two factors, say $\alpha = \alpha_1 \cdot \alpha_2$, as we proposed in~\eqn{\ref{eq:alpha hybrid}}.
Factor $\alpha_1 \simeq 8$ can be attributed to a about 8 times higher fraction of college samples in \heur compared to RW.
Factor $\alpha_2 \simeq 1.5$ is due to over-sampling smaller networks, \ie by applying stratified sampling.

Another important observation is that \heur is robust to the way we resolve target edge weight conflicts in \Sec{subsec:Step 5: Edge Conflict Resolution}. The differences between the three \heur implementations are minor - it is the application of \eqn{\ref{eq:tilde w_e(C_i)}} that brings most of the benefit.


\section{Related work}\label{sec:related}

\textbf{Graph Sampling by Exploration.} Early crawling of  P2P, OSN and WWW typically used graph traversals, mainly BFS~\cite{Ahn-WWW-07,Mislove2007,MisloveWosn08,Wilson09,Mohaisen2010} and its variants.
However, incomplete BFS introduces bias towards high-degree nodes that is unknown and thus impossible to correct in general graphs ~\cite{Achlioptas05_On_the_bias_of_traceroute_sampling,Lee-Phys-Rev-06,snowball-bias,Gjoka2010,Kurant2010}. 
Later studies followed a more principled approach based on random walks (RW) \cite{Lovasz93,AldousBookInPreparation}. The Metropolis-Hasting RW (MHRW) ~\cite{Metropolis1953,mcmc-book} removes the bias during the walk; it has been used to sample P2P networks~\cite{Stutzbach2006-unbiased-p2p,Rasti09-RDS} and OSNs~\cite{Gjoka2010}.
Alternatively, we can use RW, whose bias is known and can be corrected for \cite{Heckathorn97_RDS_introduction,Salganik2004}, thus leading to a re-weighted RW~\cite{Rasti09-RDS,Gjoka2010}.
RW was also used to sample Web~\cite{Henzinger2000},
P2P networks~\cite{Stutzbach2006-unbiased-p2p,Rasti09-RDS, Gkantsidis2004}, 
OSNs \cite{Twitter08,Rasti2008,Gjoka2010,Mohaisen2010}, and other large graphs \cite{Leskovec2006_sampling_from_large_graphs}.
It was empirically shown in \cite{Rasti09-RDS,Gjoka2010} that RW outperforms MHRW in measurement accuracy. Therefore, RW can be considered as the state-of-the-art.

Random walks have also been used to sample \emph{dynamic graphs}~\cite{Stutzbach2006-unbiased-p2p, Rasti09-RDS, Willinger09-OSN_Research}, which are outside the scope of this paper.

\smallskip
\textbf{Fast Mixing Markov Chains.} The mixing time of a random walk determines the efficiency of the sampling. 
On the practical side, the mixing time of RW in many OSNs was found larger than commonly believed~\cite{Mohaisen2010}.  
Multiple dependent random walks~\cite{Ribeiro2010}
have been used to sample disconnected and loosely connected graphs.
Random walks with jumps have been  used to sample large graphs in~\cite{Ribeiro2010a,Avrachenkov2010} and in~\cite{Leskovec2006_sampling_from_large_graphs}. 
All the above methods treat all nodes with equal importance, which is orthogonal to our technique. 

On the theoretical side, in~\cite{Boyd2004_mixing}, the authors propose a method to set edge weights that achieve the fastest mixing WRW for a given target stationary distribution. 
This technique, although related, is not applicable in our context. 
First, \cite{Boyd2004_mixing}~requires the knowledge of the graph, which makes it inapplicable to~$G$, 
yet possibly feasible in~$G^C$ (after estimating some limited information about $G^C$ as in \Sec{subsec:Estimation of Category Volumes}).
In the latter case, however, even given a perfect knowledge of~$G^C$, \cite{Boyd2004_mixing}~often assigns weight~0 to some self-loops, 
which likely makes the underlying graph~$G$ disconnected. 
Finally, and most importantly, \cite{Boyd2004_mixing} takes a target stationary distribution as input. By taking $\w^\WIS$, we will face exactly the same problems of potentially poor convergence (\Sec{subsec:Step 3: Irrelevant Categories}) and ``black holes'' (\Sec{subsec:Tiny and Unknown Categories}) as we addressed by \heur.


\smallskip
\textbf{Stratified Sampling.} Our approach builds on~\emph{stratified sampling}~\cite{Neyman1934}, a widely used technique in statistics; see
\cite{Cochran1977, Lohr2009} for a good introduction. 

A related work in a different networking problem is \cite{Duffield2005}, where threshold sampling is used to vary sampling probabilities of network traffic flows and
 estimate their volume. 

\smallskip
\textbf{Weighted Random Walks for Sampling.} Random walks on graphs with weighted edges, or equivalently reversible Markov chains \cite{Lovasz93,AldousBookInPreparation}, are well studied and heavily used in Monte Carlo Markov Chain simulations \cite{mcmc-book} to sample a state space with a specified probability distribution. However, to the best of our knowledge, WRWs have not been designed explicitly for measurements of real online systems. In the context of sampling OSNs, the closest works are \cite{Ribeiro2010a,Avrachenkov2010}. 
Technically speaking, they use WRW. But they set as their only objective the minimization of the mixing time, which makes them orthogonal and complementary to our approach, as we discussed above.

Very recent applications of weighted random walks in online social networks include~\cite{Backstrom2011, Backstrom2011_WWW}. 
\cite{Backstrom2011} uses WRW in the context of link prediction. The authors employ supervised learning techniques to set the edge weights, with the goal of increasing the probability of visiting nodes that are more likely to receive new links. 
\cite{Backstrom2011_WWW} introduces WRW-based methods to generate samples of nodes that are internally well-connected but also approximately uniform over the population. 
In both these papers, WRW is used to predict/extract something from a known graph. 
In contrast, we use WRW to estimate features of an unknown graph.


In the context of World Wide Web crawling, \emph{focused crawling} techniques \cite{Chakrabarti1999_focused,Diligenti2000_focused} have been introduced to follow web pages of specified interest and to avoid the irrelevant pages. This is achieved by performing a BFS type of sample, except that instead of fifo queue they use a priority queue weighted by the page relevancy. In our context, such an approach suffers from the same problems as regular BFS: (i)~collected samples strongly depend on the starting point, and (ii)~we are not able to unbias the sample.

\section{Conclusion}

We introduced Stratified Weighted Random Walk (\heur) - an efficient way to sample large, static, undirected graphs
via crawling and using minimal information. \heur performs a weighted random walk on the graph with weights determined by the estimation problem.
We apply our approach to measure the Facebook social graph, and we show that \heur greatly outperforms
the state-of-art sampling technique, namely the simple re-weighted random walk. 

There are several directions for future work. 
First, \heur is currently an intuitive and efficient heuristic; in future work, we plan to investigate 
the optimal solution to problems identified in this paper and compare against or improve \heur. 
Second, it may be possible to combine these ideas  with 
existing orthogonal techniques, some of which have been reviewed in Related Work, to further improve performance. 
Finally, we are interested in extending our techniques to dynamic graphs and non-stratified data.




\section*{Appendix A: Achieving Arbitrary Node Weights} \label{subsec:Proof}

Achieving arbitrary node weights by setting the edge weights in a graph~$G=(V,E)$ is sometimes impossible. For example, for a graph that is a path consisting of two nodes ($v_1 - v_2$),
it is impossible to achieve $\w(v_1)\neq\w(v_2)$.
However, it is always possible to do so, if there are self loops in each node.

\begin{observation}\label{theorem:edge_weigh_assignment_feasibility}
For any undirected graph $G=(V,E)$ with a self-loop $\{v,v\}$ at every node $v\in V$, we can achieve an arbitrary distribution of node weights $\w(v)>0,\ v\in V$, by appropriate choice of edge weights $\w(u,v)\!>\!0,\ \{u,v\}\!\in\!E$.
\end{observation}

\begin{proof}
Denote by $\w_{\min}$ the smallest of all target node weights $\w(v)$. Set $\w(u,v)=\w_{\min}/N$ for all non self-loop edges (i.e., where $u\neq v$).
Now, for every self-loop $\{v,v\}\in E$ set
$$\w(v,v)\ \ =\ \ \frac{1}{2}\left(\w(v)-\frac{\w_{\min}}{N} \cdot (\deg(v)\!-\!2)\right).$$
It is easy to check that, because there are exactly $\deg(v)\!-\!2$ non self-loop edges incident on $v$, every node~$v\in V$ will achieve the target weight $\w(v)$. Moreover, the definition of $\w_{\min}$ guarantees that $\w(v,v)>0$ for every $v\in V$.
\end{proof}

\section*{Appendix B: Estimating Category Volumes}\label{subsec: Estimating volume} 

In this section, we derive efficient estimators of the volume ratio $\est{f}^\sss{vol}_C = \frac{\vol(C)}{\vol(V)}$. 
Recall that $S\subset V$ denotes an independent sample of nodes in $G$, with replacement.

\smallskip
\noindent\textbf{Node sampling}\\
If $S$ is a uniform sample UIS, then we can write 
\begin{equation}\label{vol c 1 UIS}
	\est{f}^\sss{vol}_C\ = \ \frac{\sum_{v \in S} \deg(v)\cdot 1_{\{v\in C\}}}{\sum_{v \in S} \deg(v)},
\end{equation}
which is a straightforward application of the classic ratio estimator~\cite{Lohr2009}.

In the more general case, when~$S$ is selected using WIS, then 
we have to correct for the linear bias towards nodes of higher weights $\w()$, as follows:
\begin{eqnarray}
\label{vol c 1 WIS}  
\est{f}^\sss{vol}_C&=&  \frac{\sum_{v \in S} \deg(v)\cdot 1_{\{v\in C\}}/\w(v)}{\sum_{v \in S} \deg(v)/\w(v)}.
\end{eqnarray}
In particular, if $\w(v)\sim\deg(v)$, then
\begin{eqnarray}
\label{vol c 1 RW}  
\est{f}^\sss{vol}_C & = & \frac{1}{n} \cdot  \sum_{v \in S} 1_{\{v\in C\}}.
\end{eqnarray}
%
%

\smallskip
\noindent\textbf{Star sampling}\\
Another approach is to focus on the set of all neighbors $\mathcal{N}(S)$ of sampled nodes (with repetitions) rather than on $S$ itself, \ie to use `star sampling'~\cite{Kolaczyk2009}.
The probability that a node~$v$ is a neighbor of a node sampled from $V$ by UIS is
$$\sum_{u\in V} \frac{1}{N}\cdot 1_{\{v\in \mathcal{N}(u)\}}\ \ =\ \ \frac{\deg(v)}{N}.$$
Consequently, the nodes in $\mathcal{N}(S)$ are asymptotically equivalent to nodes drawn with probabilities linearly proportional to node degrees. By applying \eqn{\ref{vol c 1 RW}} to $\mathcal{N}(S)$, we obtain\footnote{As a side note, observe that formula \eqn{\ref{vol c 2 UIS}} generalizes the ``scale-up method'' \cite{Bernard2010} used in social sciences to estimate the size (here $|C|$) of hidden populations (\eg of drug addicts). Indeed, if we assume that the average node degree in $V$ is the same as in $C$, then $\vol(C)/\vol(V) = |C|/N$, which reduces \eqn{\ref{vol c 1 RW}} to the core formula of the scale-up method.}
\begin{equation}\label{vol c 2 UIS}
	\est{f}^\sss{vol}_C \ \ =\  \ \frac{1}{\vol(S)}  \sum_{u\in S}\sum_{v \in \mathcal{N}(u)}\!\! 1_{\{v\in C\}},
\end{equation}
where we used $|\mathcal{N}(S)|=\sum_{u\in S}\deg(u)=\vol(S)$.

In the more general case, when~$S$ is selected using WIS, then we correct for the linear bias towards nodes of higher weights $\w()$, as follows:
\begin{equation}\label{vol c 2 WIS}
	\est{f}^\sss{vol}_C\ \ =\  \ \frac{1}{\displaystyle\sum_{u\in S} \frac{\deg(u)}{\w(u)}}  \sum_{u\in S} \left(\frac{1}{\w(u)}\sum_{v \in \mathcal{N}(u)}\!\! 1_{\{v\in C\}}\right).
\end{equation}
In particular, if $\w(v)\sim\deg(v)$, then
\begin{equation}\label{vol c 2 RW}
	\est{f}^\sss{vol}_C\ \ =\  \ \frac{1}{n}  \sum_{u\in S}\left(\frac{1}{\deg(u)}\sum_{v \in \mathcal{N}(u)}\!\! 1_{\{v\in C\}}\right).
\end{equation}

\smallskip
Note that for every sampled node $v\in S$, the formulas \eqn{\ref{vol c 2 UIS}-\ref{vol c 2 RW}} exploit all the $\deg(v)$ neighbors of~$v$, whereas \eqn{\ref{vol c 1 UIS}-\ref{vol c 1 RW}} rely on one node per sample only. Not surprisingly, \eqn{\ref{vol c 2 UIS}-\ref{vol c 2 RW}} performed much better in all our simulations and implementations.

%
%

\section*{Appendix C: Relative sizes of node categories} 

Consider a scenario with only two node categories, i.e., $\mathcal{C} = \{C_1, C_2\}$.
Denote $f_1 = |C_1|/N$ and $f_2 = |C_2|/N$. The goal is to estimate $f_1$ and $f_2$ based on the collected sample~$S$.


\paragraph{UIS - Uniform independence sampling}
Under UIS, the number $X_1$ of times we select a node from~$C_1$ among~$n$ attempts follows the Binomial distribution $X_1 = Binom(f_1, n)$. Therefore, we can estimate $f_1$ as
\begin{equation}\label{eq:UIS_f_1_est}
	\hat{f}^\UIS_1\ =\ \frac{X_1}{n}  \qquad \textrm{ with } \qquad \Var(\hat{f}^\UIS_1)\ =\ \frac{f_1 f_2}{n}.
\end{equation}

\paragraph{WIS - Weighted independence sampling}

In contrast, under WIS, at every iteration the probability $\pi(v)$ of selecting a node $v$ is:
$$
\pi(v) = \left\{  \begin{array}{rl}
         \pi_1 = \frac{1}{N} \cdot\frac{w_1}{w_1f_1 + w_2f_2} & \textrm{ if } v\in C_1, \textrm{ and } \\
         \pi_2 = \frac{1}{N} \cdot\frac{w_2}{w_1f_1 + w_2f_2} & \textrm{ if } v\in C_2,
               \end{array}\right.
$$
where $w_1$ and $w_2$ are the weights $\w(v)$ of nodes in $C_1$ and $C_2$, respectively.

By applying the Hansen-Hurwitz estimator (separately for nominator and denominator), we obtain
\begin{eqnarray}\label{eq:WIS_f_1_est_naive}
\nonumber    \hat{f}^\WIS_1 &=& \frac{|\hat{C}_1|}{\hat{N}}\ =\  \frac{\sum_{v\in S} 1_{v\in C_1 }\,/\,\pi(v)}{ \sum_{v\in S} 1\,/\,\pi(v) } \\
\nonumber     &=& \frac{X_1\,/\,\pi_1}{ X_1\,/\,\pi_1\ +\ (n-X_1)\,/\,\pi_2 }\  \\
\nonumber     &=& \frac{X_1\cdot \pi_2}{X_1(\pi_2-\pi_1) + n\cdot\pi_1}\ \\
\label{eq:WIS_f_1_est}     &=& \frac{X_1\cdot \w_2}{X_1(\w_2-\w_1) + n\cdot\w_1},
\end{eqnarray}
where $X_1$ is the number of samples taken from $C_1$. Note, that to calculate $\hat{f}^\WIS_1$ we only need values $w_1$ and $w_2$, which are set by us and thus known.

Computing the variance of $\hat{f}^\WIS_1$ is a bit more challenging. We use the second-order Taylor expansions (the 'Delta method') to approximate it as follows:
\begin{eqnarray}
\nonumber \frac{\partial\hat{f}^\WIS_1}{\partial X_1} &=& \frac{nw_1w_2}{ ((w_2-w_1)X_1 +nw_1)^2 }, \quad\textrm{ and}\\
\nonumber  \Var(\hat{f}^\WIS_1) &\cong & \left(\frac{\partial\hat{f}^\WIS_1}{\partial X_1}\big(\Mean(X_1)\big)\right)^2 \Var(X_1) \\
\label{eq:WIS_f_1_var}   &=&\back\big(\ldots\big)= \frac{f_1f_2}{nw_1w_2} \cdot (f_1w_1+f_2w_2)^2 .
\end{eqnarray}
In the above derivation, we used the fact that $\Mean(X_1)=nNf_1\pi_1$ and $\Var(X_1)=nN^2f_1\pi_1f_2\pi_2$.
This comes from the fact that $X_1$ actually follows the binomial distribution $X_1 = Binom(Nf_1 \pi_1, n).$

For $w_1=w_2$, we are back in the UIS case. But this is not necessarily the optimal choice of weights. Indeed, a quick application of Lagrange multipliers reveals that
$\Var(\hat{f}^\WIS_1)$ is minimized when
\begin{equation}\label{eq:equal_total_weight}
    w_1\, f_1 = f_2\, w_2.
\end{equation}
Moreover, analogous analysis shows that \eqn{\ref{eq:equal_total_weight}} minimizes $\Var(\hat{f}_2^\WIS)$ as well. In other words, the estimators of both $f_1$ and $f_2$ have the lowest variance if the total weighted mass of $C_1$ is equal to that of $C_2$.
This implies, in expectation, equal allocation of samples between $C_1$ and $C_2$, \ie
$$n^\WIS_i\ =\ \frac{n}{|\mathcal{C}|}.$$

Finally, we can use \eqn{\ref{eq:UIS_f_1_est}}, \eqn{\ref{eq:WIS_f_1_var}} and \eqn{\ref{eq:equal_total_weight}} to calculate the gain~$\alpha$ of WIS over UIS
\begin{equation}\label{eq: gain under relative sizes}
	\alpha\ =\ \frac{1}{4f_1f_2} \quad (\geq 1).
\end{equation}
Note that we always have $\alpha\geq 1$, and $\alpha$ grows quickly with growing difference between $f_1$ and $f_2$.

\section*{Appendix D: Optimal WRW weights in Fig.~\ref{fig:tiny_example}(a)} 
Every time WRW visits the white node/category in Fig.~\ref{fig:tiny_example}(a), the next node is chosen uniformly from red and green categories. 
We stay in this selected category for $k$ rounds, where $k$ is a geometric random variable with parameter $p=w_2/(w_1\!+\!w_2)\in[0,1]$. 
Next, we come back to the white category, and reiterate the process. 
So the number $n_\sss{red}$ of times the red category is sampled is
$$n_\sss{red} = \sum_1^{Binom(0.5,n_\sss{wh})} Geom(p),$$
where $n_\sss{wh}$ is the number of visits to the white category.
Because the random variables generated by $Binom(0.5,n_\sss{wh})$ and $Geom(p)$ are independent, we can write
{\small
\begin{eqnarray}
\nonumber	\Mean[n_\sss{red}] &=& \Mean[Binom(0.5,n_\sss{wh})]\cdot \Mean[Geom(p)] \ =\ 0.5 n_\sss{wh} / p  \\
\nonumber	\Var[n_\sss{red}] &=& \Mean[Binom()]\Var[Geom()]+ \Mean^2[Geom()]\Var[Binom()]\\
\nonumber						&=& \frac{n_\sss{wh}}{4p^2} (3-2p).
\end{eqnarray}}
A possible unbiased estimator of the relative size~$f_\sss{red}$ of red category (among relevant categories) is
$$\est{f}_\sss{red} = \frac{n_\sss{red}}{n_\sss{wh}/p},$$
for which we get
\begin{eqnarray}
\nonumber \Mean[\est{f}_\sss{red}] &=& \frac{\Mean[n_\sss{red}]}{n_\sss{wh}/p} \ =\  \frac{1}{2} \quad \textrm{(unbiased)}\\
\nonumber \Var[\est{f}_\sss{red}] &=&  \frac{\Var[n_\sss{red}]}{(n_\sss{wh}/p)^2} \ =\  \frac{3-2p}{4n_\sss{wh}}.
\end{eqnarray}
This variance is expressed as a function of~$n_\sss{wh}$, and not of the total sample length~$n$. 
However, note that $n_\sss{wh}$ drops with decreasing~$p$.
Consequently, the variance $\Var[\est{f}_\sss{red}]$ (expressed as a function of $n_\sss{wh}$ or of $n$) is minimized for $p=1$, \ie for $w_1=0$ and $w_2>0$ (and $n_\sss{wh}\eq n/2$).

\end{document}